\documentclass[lettersize,journal]{IEEEtran}
\pdfoutput=1 
\IEEEoverridecommandlockouts
\usepackage[backend=bibtex,style=numeric,maxnames=2]{biblatex}
\usepackage{amsmath,amssymb,amsfonts}
\usepackage{graphicx, subcaption}
\usepackage{textcomp}
\usepackage{xcolor}
\usepackage{color,soul}
\usepackage{amsmath}
\usepackage{csquotes}

\DeclareMathOperator*{\argmin}{arg\,min}

\usepackage[english]{babel}
\usepackage{amsthm}
\usepackage{cleveref}
\usepackage{cuted}
\usepackage{color}
\newtheorem{theorem}{Theorem}

\newtheorem{lemma}[theorem]{Lemma}

\usepackage{algorithmic}
\usepackage{algorithm}

\usepackage{array}
\usepackage{stfloats}
\usepackage{verbatim}
\def\BibTeX{{\rm B\kern-.05em{\sc i\kern-.025em b}\kern-.08em
    T\kern-.1667em\lower.7ex\hbox{E}\kern-.125emX}}
\usepackage{balance}
\def\BibTeX{{\rm B\kern-.05em{\sc i\kern-.025em b}\kern-.08em
    T\kern-.1667em\lower.7ex\hbox{E}\kern-.125emX}}

\DeclareUnicodeCharacter{0394}{$\Delta$}

\begin{document}
\title{A Multi-Layered Distributed Computing Framework for Enhanced Edge Computing}

\author{Ke Ma,~\IEEEmembership{Student Member,~IEEE,}
Junfei Xie,~\IEEEmembership{Senior Member,~IEEE}

\thanks{Manuscript received Oct 2, 2024.}
\thanks{Ke Ma is with the Department of Electrical and Computer Engineering, University of California, San Diego, 	92037, USA, and also with the Department of Electrical and Computer Engineering, San Diego State University, 92115, USA. (e-mail:kem006@ucsd.edu)}
\thanks{Junfei Xie is with the Department of Electrical and Computer Engineering, San Diego State University, 92115, USA. (e-mail:jxie4@sdsu.edu)}
}

\markboth{IEEE Transactions on vehiclar Technology,~Vol.xx, No.xx, Oct~2024}%
{How to Use the IEEEtran \LaTeX \ Templates}

\maketitle

\begin{abstract}
The rise of the Internet of Things and edge computing has shifted computing resources closer to end-users, benefiting numerous delay-sensitive, computation-intensive applications. To speed up computation, distributed computing is a promising technique that allows parallel execution of tasks across multiple compute nodes.  However, current research predominantly revolves around the master-worker paradigm, limiting resource sharing within one-hop neighborhoods. This limitation can render distributed computing ineffective in scenarios with limited nearby resources or constrained/dynamic connectivity. In this paper, we address this limitation by introducing a new distributed computing framework that extends resource sharing beyond one-hop neighborhoods through exploring layered network structures. Our framework involves transforming the network graph into a sink tree and formulating a joint optimization problem based on the layered tree structure for task allocation and scheduling. To solve this problem, we propose two exact methods that find optimal solutions and three heuristic strategies to improve efficiency and scalability. The performances of these methods are analyzed and evaluated through theoretical analyses and comprehensive simulation studies. The results demonstrate their promising performances over the traditional distributed computing and computation offloading strategies.
\end{abstract}


\section{Introduction}
\IEEEPARstart{T}{he} proliferation of the Internet of Things (IoT) devices has enabled a multitude of delay-sensitive yet computation-intensive applications, such as face recognition, gaming, environment monitoring, and augmented/virtual reality \cite{chen2014vision}. The surge of these applications drives the migration of computing resources from the remote cloud to the network edge closer to end-users \cite{chen2018dynamic}. While computing at the edge offers compelling benefits, such as low latency, cost effectiveness, and improved data control and security, it also presents notable challenges. The distributed nature of edge servers, along with their inherent constraints in computing power, memory capacity, and available bandwidth when compared to the cloud, pose significant challenges to achieving high-performance edge computing \cite{shi2016edge}. 

To speed up computation, distributed computing can be employed.  
Existing distributed computing strategies typically adopt a \textit{master-worker} paradigm \cite{linderoth2000metacomputing}, where a single master node partitions and distributes the task to multiple worker nodes that are directly connected to it. 
Although the master-worker paradigm is simple to implement, it restricts resource sharing within one-hop neighborhoods. In the edge computing paradigm with constrained computing nodes, scenarios may happen where the residual resources at nearby edge servers are very limited or even less than those at the master node, rendering such master-worker based distributed computing ineffective. This challenge becomes particularly pronounced in edge networks with restricted or dynamic connectivity, such as networked airborne computing systems comprised of drones serving as edge servers \cite{lu2019toward,zhang2023exploring}.

In this paper, we overcome these challenges by exploring resources at distant servers located multiple hops away. While a similar idea has been explored in the mobile edge computing (MEC) \cite{kuang2019partial, tang2020partial, liu2023rfid, 9435782, you2016energy, barbarossa2013joint, li2018deep}  and Internet of Vehicle \cite{chen2022multihop, liu2022mobility, zhao2024asynchronous, deng2020multi} domains, where computation offloading is proposed to address users' and vehicles' computing demands, most existing studies focus on offloading tasks to a single MEC server located one or multiple hops away, with intermediate servers serving solely as relays. 
A few recent works \cite{zhang2022result, xie2022dynamic, zhang2022joint, funai2019computational} 
have considered offloading tasks to multiple servers, but they overlook the task scheduling issue addressed in this work and offer only heuristic solutions. 
To the best of our knowledge, this is the first systematic investigation into computation offloading to multiple servers across multiple hops and into the benefits of layered structures for improving distributed computing performance. Additionally, we present efficient exact solutions for optimal task allocation and scheduling, which generalize existing solutions and are applicable to a wide range of networked computing networks.

The main contributions are summarized as follows:
\begin{itemize}
\item \textit{A new multi-layered distributed computing framework:} The proposed framework explores layered network structures  to fully utilize the capacity of the entire edge computing network for enhanced system performance. It transforms the network graph into a sink tree, based on which a mixed integer programming (MIP) problem is formulated to jointly optimize task allocation and scheduling. Notably, this framework can be applied to any networked computing systems such as cloud computing, MEC, and networked airborne computing systems.  
\item \textit{Two exact methods that find optimal solutions: }To solve the optimization problem, we 
propose two exact methods, including a centralized method and a parallel enhancement of it. 
We also present an offline-online computation scheme that allows both methods to execute in real-time and 
handle dynamic and mobile networks. How these methods generalize existing solutions is also discussed.
\item \textit{Three heuristic methods for improving efficiency and scalability:} While the two exact methods offer optimality, their execution time grows rapidly as the network expands. To mitigate this challenge, we introduce two worker selection methods, as well as a genetic algorithm for efficiently finding near-optimal solutions. 
\item \textit{Comprehensive simulation studies: } 
To evaluate the performance of the proposed approaches, we conduct extensive comparison studies. Our results demonstrate that enabling resource sharing within the entire network leads to better solutions compared to those found by the traditional distributed computing and computation offloading strategies. Additionally, simulations are conducted to assess the time efficiency of the proposed approaches and the impact of their key parameters. 

\end{itemize}

It should be noted that the multi-layered distributed computing framework was initially introduced in a short conference version \cite{10622383}, which presented a different heuristic method for solving the MIP problem. This paper provides a more comprehensive and systematic investigation, featuring a set of new and rigorously analyzed methods.   

The rest of the paper is organized as follows. Sec.  \ref{sec:related_work} discusses related works. Sec. \ref{sec:modeling} describes the system model and the problem to be solved. Sec. \ref{sec:method} introduces the proposed multi-layered distributed computing framework. The two exact methods and the three heuristic methods are introduced in Sec. \ref{sec:Optimal Method} and Sec. \ref{sec:Heuristic Methods}, respectively. Sec. \ref{sec:simulation} presents the results of simulation studies. Finally, Sec. \ref{sec:conclusion} concludes the paper and discusses the future works.

\section{Related Works} \label{sec:related_work}
This section reviews existing studies related to this work. 
\subsection{Distributed Computing}
In the field of distributed computing, the master-worker paradigm has been widely used to implement parallel applications \cite{linderoth2000metacomputing, 9998108,wang2021batch,wang2021multi,zhang2023communication}. In this paradigm, multiple workers share the workload assigned by the master and communicate directly with the master. To determine the optimal task allocation, various distributed computing strategies have been proposed. For instance, traditional server-based distributed computing systems often divide the workload among workers equally or proportionally according to workers' computing power \cite{wang2021batch}. In heterogeneous systems or those with mobile compute nodes, stragglers, which are nodes with long response times, are common and can significantly degrade system performance. To mitigate the impact of stragglers, coded distributed computing techniques \cite{9998108,wang2021batch,wang2021multi} have recently become increasingly popular. These techniques leverage coding theory to introduce redundancies into computations.

Another popular paradigm is the hierarchical master-worker paradigm \cite{aida2003distributed}, which involves a supervisor process managing multiple sets of processes, each consisting of a master process and multiple worker processes. It offers several advantages over the traditional master-worker paradigm, including improved scalability and fault tolerance \cite{aida2003distributed}. Differing from these paradigms,  we investigate a multi-layer master-worker paradigm that is composed of a single master and multiple workers operating at different layers. 

\subsection{Computation Offloading}
In the computing offloading domain, most existing studies consider a single-hop single-server offloading paradigm, where tasks are offloaded from users to a single edge server within their communication range \cite{kuang2019partial, tang2020partial,liu2023rfid,li2018deep,9435782,barbarossa2013joint,you2016energy}. 
The tasks can be offloaded as a whole or partially, known as \textit{binary offloading} and \textit{partial offloading}, respectively. 
Under this paradigm, many algorithms have been designed to make the optimal offloading decisions. For instance, studies in \cite{kuang2019partial,tang2020partial} examine the scenario where tasks are offloaded from a single user to a single nearby server. \cite{liu2023rfid} extends this analysis by taking server mobility and task dependency into account.
There have also been studies that explore tasks from multiple users, optimizing not only the offloading decisions (whether to offload a task or determining the offloading ratio) but also the allocation of resources to each user. For instance, \cite{li2018deep} considers the allocation of computing resources, while \cite{9435782,barbarossa2013joint} addresses the allocation of both computing and transmission power resources. The allocation of communication resources, including time slots under the Time Division Multiple Access protocol and sub-channels under the Orthogonal Frequency-Division Multiple Access (OFDMA) protocol, is considered in \cite{you2016energy}. These problems are typically solved using numerical approaches \cite{you2016energy, barbarossa2013joint}. Reinforcement learning has also emerged as a promising tool for computation offloading \cite{9435782,li2018deep}.

In practical scenarios, it is possible that there are no (powerful) edge servers nearby for the end users. To address this limitation, researchers have started to explore multi-hop offloading, 
enabling the offloading of tasks from users to remote servers multiple hops away. 
Along this direction, existing studies mostly consider offloading tasks to a single server, as seen in \cite{hong2019multi, wang2020energy, chen2022multihop, liu2022mobility, zhao2024asynchronous, deng2020multi}. 
Additionally, the three-tier network topology comprising end users, edge servers and cloud servers \cite{qi2024bridge, liu2023joint, sun2023vehicular, almuseelem2023energy} has garnered considerable interest. In their approach, tasks are offloaded either to a nearby edge server one hop away or to a cloud server two hops away, with edge servers acting as relays. In contrast, we consider all servers reachable by users and aim to facilitate collaborative computing among them. 

There are also several works that investigate partitioning tasks into multiple parts and offloading these parts to multiple servers, which are most relevant to our study \cite{zhang2022result, zhang2022joint, funai2019computational,xie2022dynamic}. For instance, \cite{zhang2022result} investigates the joint routing and multi-part offloading for both data and result. It employs a flow model to capture data/result traffic and introduces a distributed algorithm that finds optimal solutions in polynomial time. \cite{xie2022dynamic} formulates the multi-hop offloading problem as a potential game. By dividing tasks into subtasks of equal size, each device independently decides the number of subtasks to forward or compute based on its economic utility. The study in \cite{funai2019computational} addresses the distribution of a set of tasks, partitioned from a complex application, to multiple cooperative servers that may be multiple hops away. This problem is formulated as a task assignment problem and solved by an iterative algorithm. Another relevant work is presented in \cite{zhang2022joint}, which considers a joint user association, channel allocation, and task offloading problem. It solves this problem by combining the genetic algorithm and deep deterministic policy gradient algorithm. 

Distinct from previous research, we delve into the essential benefits of layered network structures while investigating how network properties like topology and server resources affect system performance. We also address the task scheduling problem that arises when transmissions of subtasks share channels or relays, which has been overlooked by existing works. Moreover, we propose both exact and heuristic methods to solve the problem, and introduce an offline-online computation scheme to enable real-time implementation and make them adaptable to dynamic and mobile networks.

\section{System Model and Problem Description} \label{sec:modeling}
In this section, we first present the system model and then describe the problem to be solved.

\subsection{System Model}
Consider a network formed by $N+1$ edge servers, each with its own unique set of computing and communication capabilities. The servers can share resources with their neighbors either through cables in wired networks or wirelessly when they are within communication range. Additionally, a server can communicate with its one-hop neighbors simultaneously using techniques like Orthogonal Frequency Division Multiplexing \cite{hwang2008ofdm}. For simplicity, interference among the servers is not considered in this study. The entire system is supervised and managed by a control center (e.g., a software defined networking controller \cite{kreutz2014software}) to ensure that all tasks are completed efficiently and effectively.

Suppose one of the servers, referred to as \textit{master}, needs to execute a computation-intensive task that is arbitrarily decomposable, which could be generated by the server itself  or requested by a user nearby. 
To complete the task in a timely and energy-efficient manner, the master decomposes the task into subtasks and distributes them to the other servers, referred to as \textit{workers} (see Fig. \ref{fig:NetworkTopology} for an illustration). 
The master (highlighted with red) can transmit subtasks simultaneously to their neighboring servers.  However, for workers farther away, multi-hop routing is required, 
which means that each server in the network can act as a worker, a relay, or both. When a subtask arrives at a relay, it is added to a queue and processed in a first-in-first-out order. A worker will not start executing the assigned subtask until it receives the complete subtask package. When a server acts as both a worker and a relay, it can perform the relay process and execute the assigned task simultaneously. 

\begin{figure}[h]
\vspace{-0.3cm}
    \centering
    \includegraphics[width=7cm]{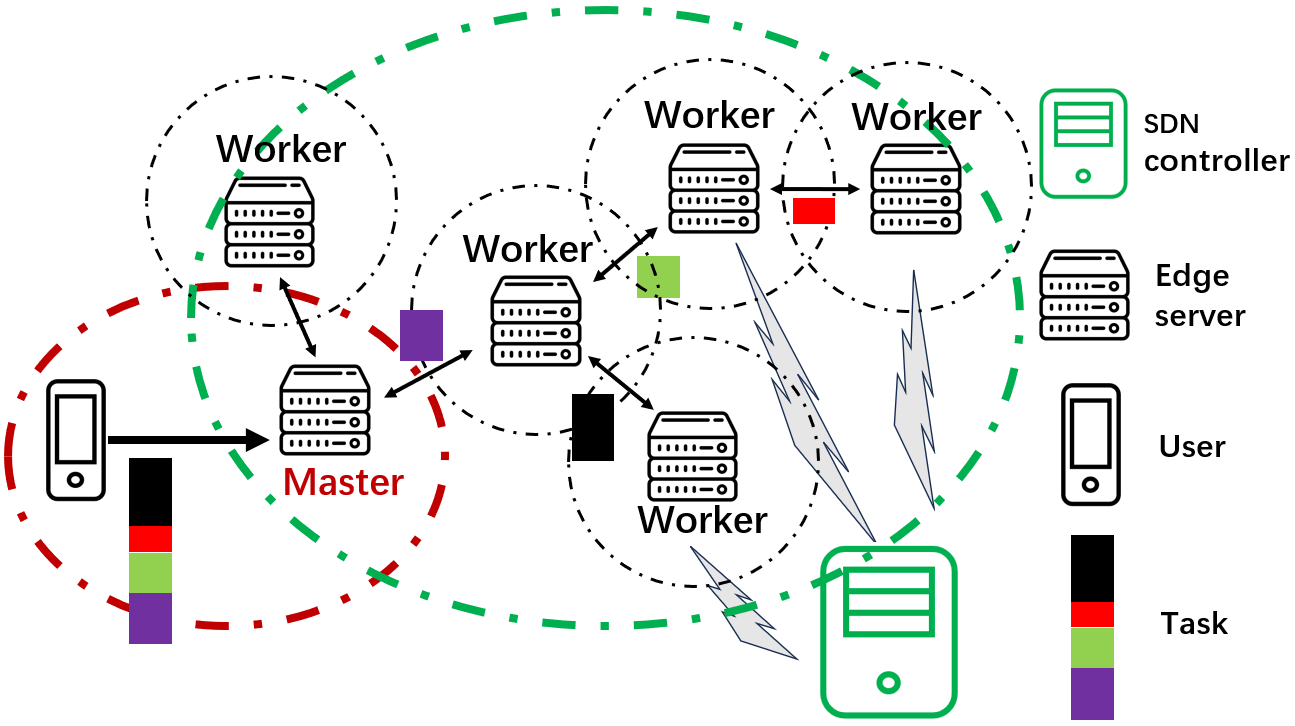}
    \caption{Network scenario.}
    \label{fig:NetworkTopology}
    \vspace{-0.3cm}
\end{figure}

In this preliminary study, we adopt several common assumptions made in existing studies \cite{8434285,li2018learning} to simplify our analysis. In particular, we assume that the network is stable with no package losses or retransmissions. Additionally, we assume that the computation result is relatively small, and hence the delay incurred in transmitting the result from workers back to the master is negligible. Under these assumptions, we model the network as follows.

\subsubsection{Network Model}
The network is modeled as a directed graph $\mathcal{G}=\{\mathcal{N}, \mathcal{E}\}$, where $\mathcal{N}=\{i|0\leq i \leq N\}$ is the set of edge servers and $\mathcal{E} = \{(i,j)|i,j\in \mathcal{N},i\neq j\}$ is the set of server-to-server communication links that connect servers that can communicate directly. 

\subsubsection{Computing model}
Let $f_i$ denote the computing capacity of server $i$, i.e., CPU-cycle frequency (GHz).  Given a task of size $y$, let $b$ denote the total number of CPU cycles required to process one task size unit. The time required for server $i$ to process this task can then be expressed by \cite{chen2022joint}:
\begin{equation}
    T_i^{comp} = \frac{yb}{f_i} \label{eq:computing_time}
\end{equation}

\subsubsection{Communication Model} 
Let $R_{i,j}$ denote the data transmission rate from server $i$ to server $j$, which we assume to be known for the sake of simplifying analysis. This rate can be approximated using the Shannon's Theory \cite{zhang2023exploring} as
 $  R_{i,j} = B_{i,j}log_2(1+\frac{s_{i,j}}{n_{i,j}})$, 
where $B_{i,j}$ is the channel bandwidth, and $s_{i,j}$ and $n_{i,j}$ represent the signal power and noise power, respectively. 

\subsubsection{Energy Consumption Model}
The energy consumed for executing a task mainly constitutes two components: energy consumed for computing and energy consumed for communication. 
The energy consumed for server $i$ to compute a task of size $y$ is given by \cite{8516294}:
    \begin{equation}
        E_i^{comp} = \gamma_i yb(f_i)^2 \label{eq:comp_energy}
    \end{equation}
where $\gamma_i$ is the effective switched capacitance that 
depends on the chip architecture of server $i$. 
The energy consumed for server $i$ to transmit a task of size $y$ to server $j$ is given by \cite{8516294}:
\begin{equation}
    E_{i,j}^{comm}=\frac{e_iy}{R_{i,j}} \label{eq:comm_energy}
\end{equation}
where $e_i$ represents the transmission power of server $i$.

\subsection{Problem Description and Analysis}
Without loss of generality, suppose the master receives a task of size $Y \in \mathrm{R}^+$ to complete. Given the computing and communication characteristics of the whole network, i.e., $\mathcal{G}$, $\{f_i, R_{i,j}, e_i,\gamma_i\}, \forall i, j \in \mathcal{N}$ 
are known, the control center aims to jointly minimize the task completion time and energy consumption by partitioning the task into small subtasks and distributing them to other servers in the network. 

Finding the optimal solution to this problem is nontrivial and challenging since it requires making decisions on several aspects, including identifying which servers the master should assign subtasks to, determining the amount of workload to be assigned to each worker, and selecting the transmission route for sending the subtask. Moreover, the order in which the subtasks should be sent by the master is also 
a crucial decision to make.

\section{Multi-Layered Distributed Computing Framework} \label{sec:method}
In this section, we present a multi-layered distributed computing framework to solve the problem described in the previous section. Motivated by the fact that a layered tree structure emerges when the master distributes tasks to other servers in the network, this framework first transforms the network graph into a sink tree and exploits this layered tree structure to find optimal task allocation and scheduling solutions. 

\subsection{Transforming Graph into a Sink Tree}
Given the network graph $\mathcal{G}$ and the characteristics of the servers and communication links forming the graph, we can find the shortest route from the master to each of the other servers in the network that takes the minimum time to transmit one bit of data. This can be achieved by defining the weight of each edge $(i,j)$ as the inverse of the associated data transmission rate, i.e., $1/R_{i,j}$, and applying the Dijkstra's algorithm \cite{dijkstra2022note} to find the most communication-efficient path. The resulting paths can then be used to construct a K-ary sink tree $\mathcal{T}$, where the master is the root node and all other servers are leaf or internal nodes reachable from the root via a unique path. This layered tree structure enables the distribution of tasks from the master to other servers in an efficient manner. 

To facilitate subsequent analysis, we re-label the nodes in the tree $\mathcal{T}$ level-by-level from the root downward, and from left to right within each level (see Fig.~\ref{fig:illustration}). Consequently, nodes in lower levels have larger indices. Let $\mathcal{I}_{l}$ denote the set of indices of nodes in level $l \in\{0,1,\ldots, H\}$, where $H$ is the height of the tree. Then, $\cup_{l=0}^H \mathcal{I}_l= \mathcal{N}$. Notably, the master (root) can transmit subtasks to its one-hop neighbors, i.e., nodes in Level $1$, simultaneously. 
However, if any one-hop neighbor has children, the subtasks assigned to them, including the one-hop neighbor, have to be transmitted one by one. 
This is because they share the same channel between the master and the one-hop neighbor, and the data arriving at the one-hop neighbor is processed in a first-in-first-out manner. Therefore, the order in which these subtasks should be sent matters.  Based on these analyses, we next formulate a joint task  allocation and scheduling problem as a mixed integer programming (MIP) model. 

\begin{figure}[b]
\vspace{-0.3cm}
    \centering
    \includegraphics[scale=0.1]{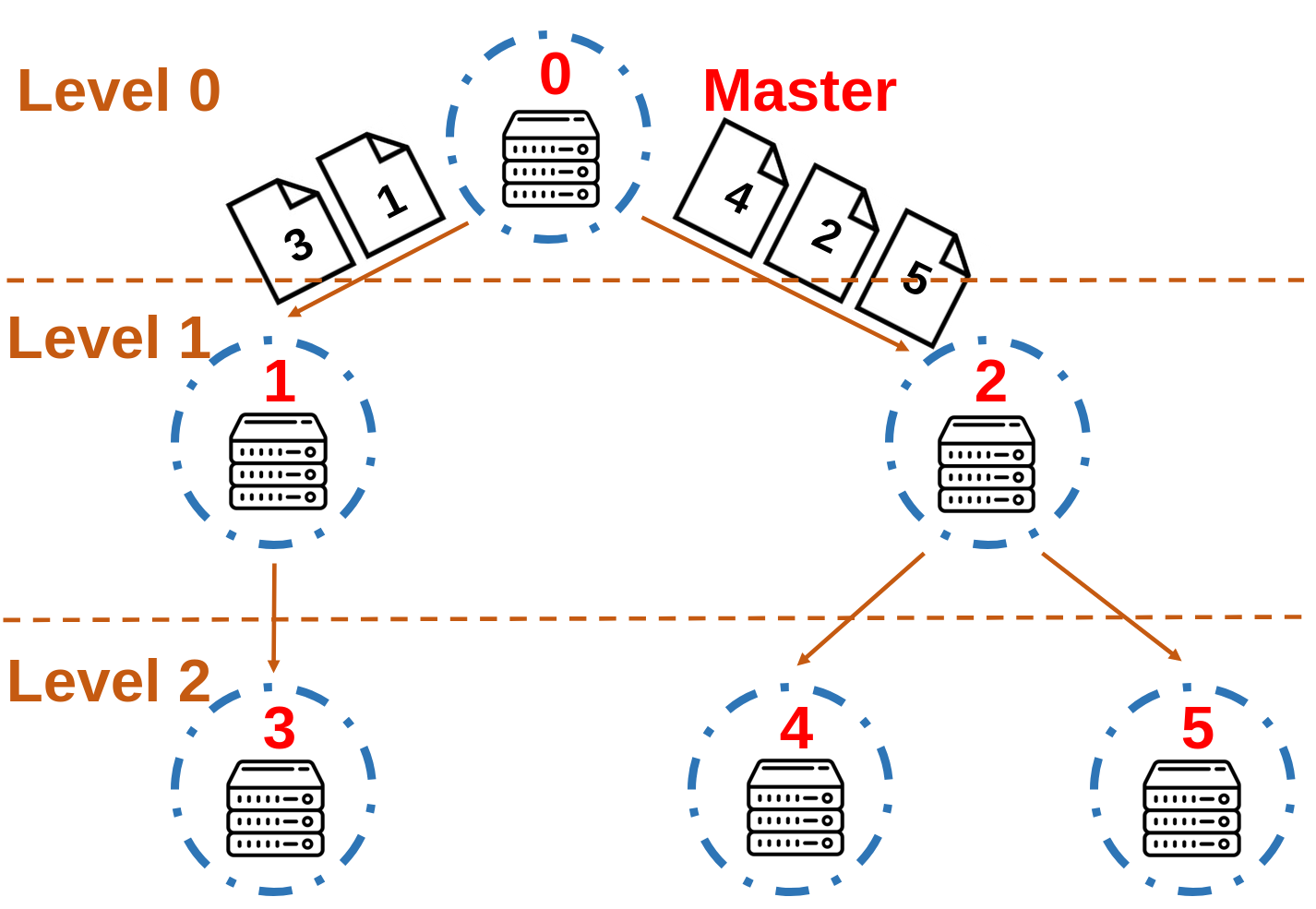}
    \caption{An example network represented by a layered tree structure. Servers' indices are highlighted in red.}
    \label{fig:illustration}
    \vspace{-0.3cm}
\end{figure}

\subsection{Mixed Integer Programming Model}
\subsubsection{Decision Variables}
To specify the computation workload allocated to each node $i\in\mathcal{N}$, we introduce decision variables $\boldsymbol{y} = \{y_0, y_1, \ldots, y_N\}$, where $y_i \in [0, Y]$ represents the size of the subtask assigned to node $i$. If $y_i = 0$, it implies that node $i$ is not assigned any workload. Note that the master may choose to execute (part of) the task locally, in which case $y_0$ would be nonzero, i.e., $y_0>0$.

To describe the offloading order for subtasks transmitted from the master to the other nodes, we introduce decision variables $\boldsymbol{o} = \{o_1, o_2, \ldots, o_N\}$, where $ o_i \in \mathcal{N}\setminus \{0\}, \forall i\in \mathcal{N}\setminus \{0\}$ and $o_i \neq o_j, \forall i, j \in \mathcal{N}\setminus \{0\}, i\neq j $. 
When $o_i > o_j$, node $i$ has a higher priority than node $j$ to receive its subtask, where $i,j\in\mathcal{N}\setminus \{0\}$ and $i \neq j$. 

\subsubsection{Objective Function}
We aim to achieve two objectives simultaneously: minimize the time spent and minimize the energy consumed by each node for executing the task. By employing a weighted sum method, we define the objective function as follows:
\begin{eqnarray}
\mathcal{J}(\boldsymbol{y}, \boldsymbol{o})  &=& \max_{i\in \mathcal{N}} w_1 T_i^{total}+ w_2E_i^{total} \label{eq:cost}\\
&=& \max_{i\in \mathcal{N}} J_i(\boldsymbol{y}, \boldsymbol{o}) \nonumber
\end{eqnarray} 
where $w_1, w_2 \geq 0$ are the weights, representing the relative importance of the two objectives. $T_i^{total}$ is the total time required for node $i$ to receive its subtask from the master and complete the assigned subtask. Note that the time required for completing the whole task is  $\max T_i^{total}$, ${i\in \mathcal{N}}$. 
$E_i^{total}$ is the total energy consumed by node $i$ during task execution. $J_i$ is introduced to denote the cost associated with node $i$. In the following sections, parentheses or subscripts may be omitted for simplicity when there is no confusion.

Next, we derive the formulas for $T_i^{total}$  and $E^{total}_i$.

\subsubsection{Time Consumption}
The task completion time for node $i$, $T_i^{total}$,  is comprised of three components: 1) time taken to transmit subtask of size $y_i$ from the master to node $i$, denoted as $T_i^{tran}$; 2) time spent waiting in the queues of relays along the path to node $i$ if any, denoted as $T_i^{wait}$; and 3) time to execute the subtask, i.e., $T_i^{comp}$. 
It is noted that the waiting time $T_i^{wait}$ is impacted by the task sizes assigned to other nodes and the offloading order, which complicates the optimization problem considered in this study.

To obtain the transmission time $T_i^{tran}$, we introduce the notation $\boldsymbol{p}_i$ to denote the sequence of nodes that lie on the path from the master to node $i$, and the notation $p_{ik}$ to denote the $k$-th node in the sequence, where $1 \leq k \leq |\boldsymbol{p}_i|$, $p_{i1} = 0$ and $p_{i|\boldsymbol{p}_i|} = i$. $|\cdot|$ finds the cardinality of a set. $T_i^{tran}$ can then be expressed by:
\begin{equation}
    T_{i}^{tran} =  \begin{cases}
    0, & \text{if } i = 0 \\
 \sum_{k = 1}^{|\boldsymbol{p}_i|-1}\frac{y_i}{R_{{p_{ik},p_{i(k+1)}}}}, & \text{else}
    \end{cases}\label{eq:transmission_time}
\end{equation}

Let's now consider the waiting time $T_i^{wait}$. 
Let $\mathcal{A}_t$ denote the full set of nodes in the $t$-th subtree of the master, where $t\in \mathcal{I}_1$, and $\cup_{t\in \mathcal{I}_1} \mathcal{A}_t = \mathcal{N}\setminus \{0\}$.  
Additionally, define $\mathcal{B}_i = \{j | o_j > o_i, i,j \in \mathcal{A}_t, i\neq j\}$ as the set of nodes whose subtasks will be transmitted before node $i$. Note that if nodes $i$ and $j$ belong to different subtrees, i.e., $i \in \mathcal{A}_t$ while $j\notin \mathcal{A}_t$, the subtask for node $j$ is transmitted using a different channel that is orthogonal to the one used for node $i$, and hence node $i$ does not need to wait for node $j$'s subtask to be transmitted even if $o_j > o_i$. Based on these definitions, we can then express the waiting time as follows:
\begin{equation}
T_i^{wait} = \begin{cases} 0, &\text{if~} i = 0 \text{ or } \mathcal{B}_i = \emptyset \\
\sum_{j \in \mathcal{B}_i} \sum_{k=1}^{|\boldsymbol{p}_s |-1} \frac{y_j}{R_{p_{sk},p_{s(k+1)}}}, &\text{else} \end{cases}
\label{eq:waiting_time}
\end{equation}
where $\boldsymbol{p}_s  = \boldsymbol{p}_i \cap \boldsymbol{p}_j$.

Based on \eqref{eq:computing_time}, \eqref{eq:transmission_time}, and \eqref{eq:waiting_time}, we then have
\begin{equation}
T_i^{total}= T_i^{trans} + T_i^{wait} + T_i^{comp}\label{eq:total_time}
\end{equation}

\subsubsection{Energy Consumption} 
With $T_i^{total}$ and \eqref{eq:comp_energy}-\eqref{eq:comm_energy}, the energy consumption $E_i^{total}$ can then be expressed by:
\begin{equation}
E_i^{total}= E_i^{comp} + \sum_{j\in \mathcal{C}_i}E_{i,j}^{comm}
\label{eq:total_energy}
\end{equation}
In the above equation, $\mathcal{C}_i$ is the set of children of node $i$, whose subtasks will be relayed by node $i$.

\subsection{Problem Formulation}
Mathematically, the multi-objective optimization problem can be formulated as follows:
\begin{equation}
\begin{aligned} 
\mathcal{P}_0: ~   &\min_{\boldsymbol{y}, \boldsymbol{o}} \mathcal{J}(\boldsymbol{y}, \boldsymbol{o})\\
    s.t.~& \sum\limits_{i=0}^N y_i = Y   & C1\nonumber\\
    & 0 \leq y_i \leq Y, \forall i \in \mathcal{N} & C2\nonumber\\
   &  o_i \in \mathcal{N}\setminus \{0\}, \forall i\in \mathcal{N}\setminus \{0\} & C3\nonumber\\
   & o_i \neq o_j, \forall i, j \in \mathcal{N}\setminus \{0\}, i\neq j  & C4\nonumber
\end{aligned}
\end{equation}

Constraint $C1$ ensures that the total assigned workload sums up to the total task size $Y$. Constraints $C2$-$C4$ guarantee that each decision variable takes on valid values.

\section{Joint Optimization of Task Allocation and Scheduling} \label{sec:Optimal Method}
In this section, we introduce two exact methods to find the optimal solution to the joint task allocation and scheduling problem $\mathcal{P}_0$.

\subsection{Centralized MILP-based Optimization (CMO)} \label{sec:ex search}

\subsubsection{Algorithm Description}
It is noted that $\mathcal{P}_0$, which aims to minimize the maximum cost of individual nodes, is a minmax optimization problem. Hence, we can convert it into an equivalent mixed integer linear programming (MILP) problem by introducing an auxiliary variable $z$ as follows: 
\begin{equation}
\begin{aligned} 
 \mathcal{P}_1: ~   &\min_{\boldsymbol{y}, \boldsymbol{o}, z} \quad z \\
    s.t.~& z \geq J_i (\boldsymbol{y}, \boldsymbol{o}), \forall i \in \mathcal{N} \\ 
    &C1 - C4 
    \label{eq:minproblem}
\end{aligned}
\end{equation}
Then, the minimum cost $\mathcal{J}^{\star} = z^{\star}$, where $z^{\star}$ is the minimum value of $z$ found by solving $\mathcal{P}_1$.

Problem $\mathcal{P}_1$ can be further decomposed into two subproblems. The first subproblem aims to optimize the task allocation $\boldsymbol{y}$, given a particular offloading order denoted as $\boldsymbol{o} =\boldsymbol{o}_k $: 
\begin{equation}
\begin{aligned} 
\mathcal{P}_1^{(a)}: ~   &\min_{\boldsymbol{y},z} \quad z \\
    s.t.~& z \geq  J_i(\boldsymbol{y}, \boldsymbol{o}_k), \forall i \in \mathcal{N} 
    \nonumber \\
    & C1 - C2 \nonumber
\end{aligned}
\end{equation}
Denote the optimal solution to problem $\mathcal{P}_1^{(a)}$ at $\boldsymbol{o} =\boldsymbol{o}_k $ as $\{{\boldsymbol{y}^*}(\boldsymbol{o}_k), {z}^*(\boldsymbol{o}_k)\}$. 
The second subproblem aims to optimize the offloading order $ \boldsymbol{o}$: 
\begin{equation}
\begin{aligned} 
\mathcal{P}_1^{(b)}: ~   &\min_{\boldsymbol{o}_k} \quad {z}^*(\boldsymbol{o}_k)  \nonumber 
\end{aligned}
\end{equation}

Now let's consider subproblem $\mathcal{P}_1^{(a)}$, which can be solved using Lagrange multipliers \cite{boyd2004convex}. Particularly, the Lagrangian function can be defined as follows:
\begin{align}
\begin{split}
\mathcal{L}(\boldsymbol{y}, z, &\boldsymbol{\lambda}, \mu) = z +\sum_{i=0}^{N} \lambda_i\left[J_i(\boldsymbol{y}, \boldsymbol{o}_k) - z\right] + \mu\left(\sum_{i=0}^N y_i - Y\right), \nonumber
\end{split}
\end{align}
where $\boldsymbol{\lambda} = [\lambda_0, \lambda_1, \ldots, \lambda_N]$ and $\mu$ are Lagrangian multipliers. $\lambda_i \ge 0$, $\forall i \in \mathcal{N}$.  Define \[g(\boldsymbol{\lambda}, \mu) = \min\limits_{\boldsymbol{y},z} \mathcal{L}(\boldsymbol{y}, z, \boldsymbol{\lambda},\mu).\] 
The dual optimization problem is then constructed as follows:
\begin{eqnarray}
  &\max\limits_{\boldsymbol{\lambda}, \mu}&g(\boldsymbol{\lambda}, \mu) \label{eq:dual}\\
  &s.t.& \boldsymbol{\lambda} \ge 0 \nonumber
\end{eqnarray}

As the objective function and the inequality constraints in our problem are convex, and the equality constraints are affine and strictly feasible, Slater's condition \cite{auslender2000lagrangian} is satisfied and the strong duality holds. That means the optimal value of the primal problem $\mathcal{P}_1^{(a)}$ is equal to the optimal value of its dual problem \eqref{eq:dual}.
The optimal solution to $\mathcal{P}_1^{(a)}$ can then be found by solving the following equation set, known as the Karush-Kuhn-Tucker (KKT) conditions \cite{luo2006introduction}:

\begin{equation}
\left\{
\begin{array}{l}
    \frac{\partial}{\partial y_i}\mathcal{L}(\boldsymbol{y}, z,\boldsymbol{\lambda}, \mu) = 0, ~ \forall i \in \mathcal{N} \\
    \frac{\partial}{\partial z}\mathcal{L}(\boldsymbol{y}, z,\boldsymbol{\lambda}, \mu)  = 0 \\
    \sum_{i=0}^N y_i = Y \\
    \lambda_i (J_i(\boldsymbol{y}) - z) = 0, ~\forall i \in \mathcal{N} \\
    J_i(\boldsymbol{y}) - z \leq 0, ~\forall i \in \mathcal{N} \\
    \lambda_i \geq 0, ~\forall i \in \mathcal{N}
\end{array} \label{eq:general solution}
\right.
\end{equation}

To solve problem $\mathcal{P}_1^{(b)}$, we can use exhaustive search. This involves computing the cost $z^*(\boldsymbol{o}_k)$ for each possible offloading order $\boldsymbol{o}_k$ and selecting the one that yields the smallest cost. 
However, as $\boldsymbol{o}_k$ can take $N!$ possible values, evaluating each possible value is time-consuming. 
A significant reduction in the number of possible values to evaluate can be achieved by exploiting the parallelism in sending subtasks belonging to different subtrees of the master. 
Specifically, the offloading orders for nodes in any subtree $\mathcal{A}_t$ are independent of those in any other subtree $\mathcal{A}_{t'}$, where $t,t'\in\mathcal{I}_1$ and $t\neq t'$. Therefore, the number of possible values of $\boldsymbol{o}_k$ that need to be evaluated can be reduced to $\prod_{t\in \mathcal{I}_1} |\mathcal{A}_t|!$. Algorithm \ref{Algo:Exhaustive Search Algorithm} summarizes the procedure of the proposed approach, named the centralized MILP-based optimization (CMO).

\begin{algorithm} 
    \caption{\textsc{CMO}($\mathcal{T}$, $Y$)}
    \label{Algo:Exhaustive Search Algorithm}
    \begin{algorithmic}[1]
        \FOR{each $\boldsymbol{o}_k, k \in \{1,2,\ldots, \prod_{i\in \mathcal{I}_1} |\mathcal{A}_i|!\}$}
            \STATE Find $\{\boldsymbol{y}^*(\boldsymbol{o}_k),z^*(\boldsymbol{o}_k)\}$ by solving equation set (\ref{eq:general solution});
        \ENDFOR
        \STATE $\boldsymbol{o}^{\star} \gets \argmin_{\boldsymbol{o}_k} z^*(\boldsymbol{o}_k)$, $z^{\star} \gets z^*(\boldsymbol{o}^{\star})$, $\boldsymbol{y}^{\star} \gets \boldsymbol{y}^*(\boldsymbol{o}^{\star})$
        \RETURN $z^{\star},\boldsymbol{y}^{\star},\boldsymbol{o}^{\star}$
    \end{algorithmic}
\end{algorithm}

\subsubsection{Computational Complexity Analysis}
As the equation set \eqref{eq:general solution} involves $2N+4$ unknown variables, solving it requires $O(N^3)$ amount of time in the worst case \cite{greub2012linear}. The computational complexity of CMO is hence $O(\prod_{t\in \mathcal{I}_1} |\mathcal{A}_t|! N^3)$, with a worst-case complexity of $O(N!N^3)$.

\subsection{Parallel MILP-based Optimization (PMO)}
\subsubsection{Algorithm Description}
The parallelism involved in sending subtasks belonging to different subtrees of the master can be further harnessed to greatly enhance efficiency. Specifically, the key idea is to decompose problem $\mathcal{P}_1$ alternatively into two different subproblems. The first subproblem optimizes the task allocation and scheduling for nodes within each subtree, which can be solved in parallel. The second subproblem optimizes the total workload assigned to each subtree. 

Mathematically, let $Y_t$ be the total workload assigned to nodes within the $t$-th subtree of the master, i.e., $Y_t= \sum_{i\in \mathcal{A}_t}y_i$, $t\in\mathcal{I}_1$. Additionally, let $\boldsymbol{y}_t = \{y_i| i\in \mathcal{A}_t\}$ and $\boldsymbol{o}_t = \{o_i | i\in \mathcal{A}_t\}$ represent the decision variables associated with nodes within the subtree $\mathcal{A}_t$. Then, the first subproblem can be formulated as follows: 
\begin{equation}
\begin{aligned} 
 \mathcal{P}^{(a')}_1: ~   &\min_{\boldsymbol{y}_t, \boldsymbol{o}_t, z_t} \quad z_t \\
    s.t.~& z_t \geq J_i(\boldsymbol{y}_t,\boldsymbol{o}_t), \forall i \in \mathcal{A}_t  \\ 
    & \sum_{i\in \mathcal{A}_t} y_i = Y_t \nonumber \\
    & 0 \leq y_i \leq Y_t, \forall i\in\mathcal{A}_t \nonumber\\
    &  o_i \in \mathcal{N}\setminus \{0\}, \forall i\in \mathcal{A}_t \nonumber\\
   & o_i \neq o_j, \forall i, j \in \mathcal{A}_t, i\neq j  \nonumber
\end{aligned}
\end{equation}
Since tasks assigned to different subtrees can be transmitted simultaneously, this problem can be solved independently and in parallel for different subtrees.

Suppose given $Y_t$, the optimal solution to problem $\mathcal{P}_1^{(a')}$ for subtree $\mathcal{A}_t$ is $\{\bar{z}_t(Y_t), \bar{\boldsymbol{y}}_t(Y_t), \bar{\boldsymbol{o}}_t(Y_t)\}$. The second subproblem aims to optimize the workload assigned to each subtree as well as to the master, denoted as $\boldsymbol{\mathcal{Y}}= \{Y_t| t\in \mathcal{I}_1\}\cup\{y_0\}$, which can be mathematically formulated as follows:
\begin{equation}
\begin{aligned}
 \mathcal{P}^{(b')}_1: ~   &\min_{\boldsymbol{\mathcal{Y}}} \quad z \\
    s.t.~& z \geq \bar{z}_t(Y_t), \forall t \in \mathcal{I}_1\\
    & z \geq J_0(\boldsymbol{\mathcal{Y}}) \\
    & y_0 + \sum_{t \in \mathcal{I}_1} Y_t = Y \nonumber
\end{aligned}
\end{equation}
where $J_0(\boldsymbol{\mathcal{Y}}) = w_1T_0^{total} + w_2 E_0^{total}= w_1 \frac{y_0b}{f_0} + w_2 \left[\gamma_0 y_0 b(f_0)^2 + \sum_{t \in \mathcal{I}_1} \frac{e_0 Y_t}{R_{0, t}}\right]$. Of note, this subproblem can be conceptualized by abstracting each subtree as a single node. Therefore, $\mathcal{T}$ is abstracted as a one-layer tree (see Fig. \ref{fig:Greedy Flat}). The optimization of the workload $Y_t$ assigned to each abstracted node (subtree) thus does not require consideration of the offloading order.
Then, the optimal solution to $\mathcal{P}^{(b')}_1$, denoted as $Y^*_t$, $\forall t\in \mathcal{I}_1$, can be used to derive the optimal solution to the original problem $\mathcal{P}_1$. Particularly, $\boldsymbol{y}^{\star} = \{\bar{\boldsymbol{y}}_t(Y^*_t)| t\in \mathcal{I}_1\}$, and the optimal offloading order for nodes within each subtree $t$ is given by  $\bar{\boldsymbol{o}}_t(Y^*_t)$.

\begin{figure}[t]
\vspace{-0.3cm}
    \centering
    \includegraphics[scale=0.45]{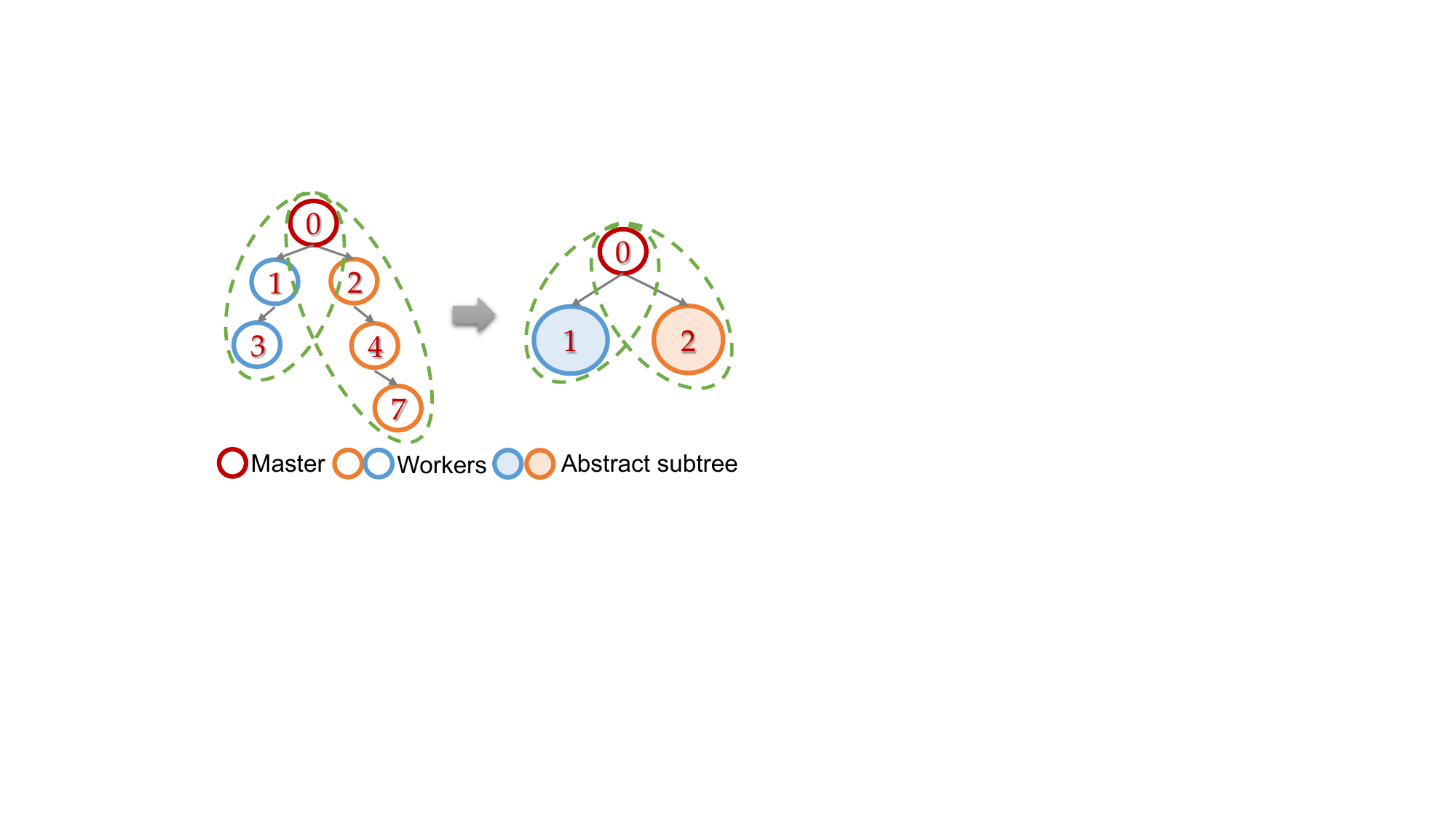}
    \caption{Illustration of how a network tree can be abstracted as a one-layer tree. }
    \label{fig:Greedy Flat}
    \vspace{-0.3cm}
\end{figure}

Solving problem $\mathcal{P}_1^{(a')}$ given $Y_t$ is relatively straightforward. However, directly addressing subproblem $\mathcal{P}_1^{(b')}$ is challenging because $Y_t$ is continuous, and obtaining $\bar{z}_t(Y_t)$ in the constraints requires solving subproblem $\mathcal{P}_1^{(a')}$. Before we proceed with our approach to solving these subproblems, we present the following lemma and theorem, which allow for their simplification.

\begin{lemma} \label{lemma:lm1}
Given $\mathcal{T}$ and $Y$, for an arbitrary offloading order $\boldsymbol{o}_k $, suppose $\{{\boldsymbol{y}^*}(\boldsymbol{o}_k), {z}^*(\boldsymbol{o}_k)\}$ is an optimal solution to problem $\mathcal{P}_1^{(a)}$. Then, $\{\frac{Y'}{Y}{\boldsymbol{y}^*}(\boldsymbol{o}_k), \frac{Y'}{Y}{z}^*(\boldsymbol{o}_k)\}$ is an optimal solution to $\mathcal{P}^{(a)}_1$ when the task size is changed to $Y'$.

\end{lemma}

\begin{proof}
As detailed in Sec. \ref{sec:Optimal Method}, the optimal solution to $\mathcal{P}_1^{(a)}$ can be found by solving the equation set \eqref{eq:general solution}. Suppose $\{z^*(\boldsymbol{o}_k), \boldsymbol{y}^*(\boldsymbol{o}_k), \boldsymbol{\lambda}^*(\boldsymbol{o}_k), \mu^*(\boldsymbol{o}_k)\}$ is the obtained optimal solution for task size $Y$. Note that the cost of each node $J_i$ can be expressed as a linear combination of the task assignments $\{y_0, y_1, \ldots, y_N\}$, i.e., $J_i = a_{i,0} y_0+ a_{i,1} y_1+...+a_{i,{|\mathcal{N}|}} y_{|\mathcal{N}|}$, where $\{a_{i,0},...,a_{i, |\mathcal{N}|}\}$  are constants that depend on network characteristics.
Hence, equation set (\ref{eq:general solution}) can be simplified as:
\begin{equation}
\left\{
\begin{array}{l}
    \sum_{i=0}^{|\mathcal{N}|}\lambda_i\frac{\partial (a_{i,0} y_0+ a_{i,1} y_1 +...+a_{i,{|\mathcal{N}|}} y_{|\mathcal{N}|})}{\partial y_i} +\mu = 0, ~ \forall i \in \mathcal{N} \\
    1 - \sum_{i=0}^{|\mathcal{N}|}\lambda_i  = 0 \\
    \sum_{i=0}^N y_i = Y \\
    \lambda_i \left(a_{i,0} y_0+ a_{i,1} y_1+...+a_{i,{|\mathcal{N}|}} y_{|\mathcal{N}|} - z\right) = 0, ~\forall i \in \mathcal{N} \\
    a_{i,0} y_0 + a_{i,1} y_1+...+a_{i,{|\mathcal{N}|}} y_{|\mathcal{N}|} - z \leq 0, ~\forall i \in \mathcal{N} \\
    \lambda_i \geq 0, ~\forall i \in \mathcal{N}
\end{array} 
\right. \nonumber
\end{equation}

When the task size is changed to $Y'$, the solution $\{\frac{Y'}{Y}z^*(\boldsymbol{o}_k), \frac{Y'}{Y}\boldsymbol{y}^*(\boldsymbol{o}_k), \boldsymbol{\lambda}^*(\boldsymbol{o}_k), \frac{Y'}{Y}\mu^*(\boldsymbol{o}_k)\}$ satisfies the above equation set. This indicates that it is an optimal solution to  $\mathcal{P}_1^{(a)}$ for task size $Y'$. With this, the proof is now complete.
\end{proof}

Lemma \ref{lemma:lm1} leads directly to the following theorem. 
\begin{theorem} \label{Thm:factor}
Given $\mathcal{T}$ and $Y$, suppose $\{z^{\star}, \boldsymbol{y}^{\star}, \boldsymbol{o}^{\star}\}$ is the optimal solution to problem $\mathcal{P}_1$. Then, $\{\frac{Y'}{Y}z^{\star}, \frac{Y'}{Y}\boldsymbol{y}^{\star}, \boldsymbol{o}^{\star}\}$ is the optimal solution to $\mathcal{P}_1$ when the task size is changed to $Y'$.
\end{theorem}

\begin{proof}
Theorem \ref{Thm:factor} can be directly derived from Lemma \ref{lemma:lm1} when $\boldsymbol{o}_k = \boldsymbol{o}^{\star}$.
\end{proof}

The proportionality property described in Theorem \ref{Thm:factor} infers that, given the optimal solution to $\mathcal{P}_1^{(a')}$ for any $Y'_t$, i.e., $\{\bar{z}_t(Y'_t), \bar{\boldsymbol{y}}_t(Y'_t), \bar{\boldsymbol{o}}_t(Y'_t)\}$, the subproblem $\mathcal{P}_1^{(b')}$ can be simplified to:
\begin{equation}
\begin{aligned}
 \mathcal{P}^{(c')}_1: ~   &\min_{\boldsymbol{\mathcal{Y}}} \quad z \\
    s.t.~& z \geq \frac{Y_t}{Y'_t}\bar{z}_t(Y'_t), \forall t \in \mathcal{I}_1\\
     & z \geq J_0(\boldsymbol{\mathcal{Y}}) \\
    & y_0 + \sum_{t \in \mathcal{I}_1} Y_t = Y \nonumber
\end{aligned}
\end{equation}
Since $Y'_t$ and $\bar{z}_t(Y'_t)$ are known (by solving $\mathcal{P}_1^{(a')}$), $\mathcal{P}^{(c')}_1$ is now straightforward to solve. 

Next, we describe our approach to solve subproblems $\mathcal{P}^{(a')}_1$  and $\mathcal{P}^{(c')}_1$. Particularly, for $\mathcal{P}^{(a')}_1$, we can solve it by leveraging the CMO algorithm (Algorithm \ref{Algo:Exhaustive Search Algorithm}). For each subtree $\mathcal{A}_t$, $t\in\mathcal{I}_1$, we run the CMO algorithm on the tree formed by $\mathcal{A}_t$ as well as the master (as highlighted by the green dashed circle in Fig. \ref{fig:Greedy Flat}), denoted as $\mathcal{T}_t$, where the input $Y$ can take any value. Suppose the output generated by CMO for $\mathcal{T}_t$ is denoted as $\{\tilde{z}_t, \tilde{\boldsymbol{y}}_t, \tilde{\boldsymbol{o}}_t\}$, where $\tilde{\boldsymbol{y}}_t=\{\tilde{y}_i|i\in\mathcal{A}_t\} \cup \{\tilde{y}_0\}$  specifies the tasks allocated to the nodes within $\mathcal{T}_t$. Then we let $Y'_t = \sum_{i \in \mathcal{A}_t} \tilde{y}_i$, $\bar{z}_t(Y'_t) = \max \{ J_i(\tilde{\boldsymbol{y}}_t, \tilde{\boldsymbol{o}}_t) | i \in \mathcal{A}_t \} $. Given $Y'_t$ and $\bar{z}_t(Y'_t)$ for each $t\in \mathcal{I}_1$, we can then solve the subproblem $\mathcal{P}^{(c')}_1$ using a commercial solver, such as Gurobi \cite{pedroso2011optimization} and CVX \cite{guimaraes2015tutorial}. 

Denote the optimal solution to subproblem $\mathcal{P}^{(c')}_1$ as $\boldsymbol{\mathcal{Y}}^* = \{{Y}^*_t|t\in\mathcal{I}_1\} \cup \{{y}^*_0\}$. 
The optimal solution to the original problem $\mathcal{P}_1$ can then be derived as:
\begin{eqnarray}
y^{\star}_i &=& \frac{Y^*_t}{Y'_t} \tilde{y}_i, \forall  i \in \mathcal{A}_t, t\in \mathcal{I}_1 \label{eq: original_task}\\
z^{\star} &=& \max \left\{ \max_{t\in \mathcal{I}_1}\frac{Y^*_t}{Y'_t} \bar{z}_t(Y'_t),J_0(\boldsymbol{\mathcal{Y}}^*)\right\} \label{eq: original_cost}
\end{eqnarray}
$\tilde{\boldsymbol{o}}_t$, $t \in \mathcal{I}_1$, specifies the optimal offloading order for subtasks assigned to each node within each subtree $\mathcal{A}_t$, where subtasks for different subtrees can be transmitted simultaneously. 

Algorithm \ref{Algo:Greedy Flat Algorithm} summarizes the procedure of the parallel MILP-based optimization (PMO) method.

\begin{algorithm} [h]
    \caption{\textsc{PMO}($\mathcal{T}$, $Y$)}
    \label{Algo:Greedy Flat Algorithm}
    \begin{algorithmic}[1] 
        \FOR {each $t \in \mathcal{I}_1$}  
        \STATE $\{\tilde{z}_t, \tilde{\boldsymbol{y}}_t, \tilde{\boldsymbol{o}}_t\} \gets$ {\textsc{CMO}($\mathcal{T}_t$, $Y$)} in parallel;\label{inalgo:greedyflat}
        \ENDFOR
        \STATE Find $\boldsymbol{\mathcal{Y}}^*$ by solving $\mathcal{P}^{(c')}_1$; \label{inalgo:algo2line6}
        \STATE Calculate $\boldsymbol{y}^\star, z^\star$ using \eqref{eq: original_task} and \eqref{eq: original_cost}, respectively; 
        \RETURN  $z^{\star},\boldsymbol{y}^{\star}$, and $\{\tilde{\boldsymbol{o}}_t| t\in \mathcal{I}_1\}$;
    \end{algorithmic}
\end{algorithm}


\subsubsection{Computational Complexity Analysis}
Since subtree $\mathcal{T}_t$, $t\in\mathcal{I}_1$, contains $|\mathcal{A}_t|+1$ nodes,  
CMO$(\mathcal{T}_t, Y)$ requires $O((|\mathcal{A}_t|+1)!(|\mathcal{A}_t|+1)^3)$ time to execute.  
The complexity of PMO is $O(\max_{t\in \mathcal{I}_1} (|\mathcal{A}_t|+1)!(|\mathcal{A}_t|+1)^3)$ with a worst-case complexity of $\mathcal{O}(N!N^3)$.

\subsection{Offline-Online Computation}
Despite the fact that the computational complexity of CMO and PMO grows rapidly as the network expands, both can be executed in real-time  
by transferring the majority of the computations offline. This can be achieved by leveraging the proportionality property presented in Theorem \ref{Thm:factor}. Particularly, for any task size $Y$, we can execute Algorithm \ref{Algo:Exhaustive Search Algorithm} \textit{offline} to derive a baseline optimal solution $\{z^{\star}, \boldsymbol{y}^{\star}, \boldsymbol{o}^{\star}\}$. Then, during \textit{online} computations, upon receiving a new task $Y'$, we can readily compute the associated optimal solution in real-time by scaling the baseline with a factor $\frac{Y'}{Y}$, i.e.,  $\{\frac{Y'}{Y}z^{\star}, \frac{Y'}{Y}\boldsymbol{y}^{\star}, \boldsymbol{o}^{\star}\}$. 

This offline-online computation scheme also equips CMO and PMO with the ability to handle dynamic networks with time-varying network characteristics. One approach to deploying them in dynamic networks is to periodically execute Algorithm \ref{Algo:Exhaustive Search Algorithm} to update the baseline solution with the latest network information. Alternatively, the baseline solution can be updated when significant network changes occur, such as alterations in the network topology.

\section{Heuristic Methods} \label{sec:Heuristic Methods}
In this section, we introduce three heuristic methods to further speed up the computation. 

\subsection{Worker Selection}
Through simulation studies, as presented in Sec. \ref{sec:simulation}, we find that the solutions produced by CMO and PMO typically improve as more workers participate in computations. However, 
the rate of performance improvement diminishes as the network size reaches a certain threshold. 
This observation inspires us to consider selecting a subset of workers that contribute the most to performance improvement. Next, we introduce two worker selection methods: 1) a \textit{node pruning} (NP) strategy, and 2) a \textit{level pruning} (LP) strategy. These methods can be applied either individually or in combination. When PMO is utilized, they can be employed to prune each subtree $\mathcal{T}_t$ before executing Line 2 in Algorithm \ref{Algo:Greedy Flat Algorithm}.

\subsubsection{Node Pruning (NP)}
The key idea of NP is to ``prune" nodes that are too costly to use. Specifically, this strategy evaluates each node one by one. For a given node $i$, it estimates the cost of using this node by performing partial offloading \cite{8516294}, which finds the optimal task partition between the master and node $i$ exclusively. 
The obtained cost, denoted as $z^p_{i}$, is then compared with the cost of local computing, i.e., the cost of processing the entire task $Y$ at the root node, denoted as $z^{(0)}$, which can be obtained by running $\text{CMO}(\mathcal{I}_0, Y)$. If the cost reduction, measured by $\frac{z^{(0)} - z^p_i}{z^{(0)}}$, exceeds a predefined threshold $\theta_p$, node $i$ is selected; otherwise, it is ``pruned". Here, ``prune" means that no workload is assigned to the node. If the node is a leaf, it is removed from the tree. However, if it is an intermediate node with unpruned children, it remains and only acts as a relay.

\subsubsection{Level Pruning (LP)}
The key idea of LP is to trim nodes that are excessively distant from the master node, whose computing resources are too costly to use considering the significant communication costs. Specifically, this strategy evaluates the top $\xi$ levels of the original network tree, removing levels from $\xi+1$ to $H$. The resulting tree, denoted as $\mathcal{T}^{\xi}$, satisfies $\mathcal{T}^{\xi} = \mathcal{T} \setminus \cup_{l=\xi}^H \mathcal{I}_l$.

\subsection{Genetic Algorithm}
The worker selection methods allow us to reduce workers but may prune nodes that could significantly improve system performance. Here, we introduce a genetic algorithm (GA) \cite{mirjalili2019genetic} that allows us to evaluate large networks. It involves two phases: \textit{initialization} and \textit{training}. In the initialization phase, a population set $\boldsymbol{O} = \{\boldsymbol{o}_k\}$ is first randomly generated, which consists of $P$ offloading orders (chromosomes). The corresponding optimal task partition $\boldsymbol{y}^*(\boldsymbol{o}_k)$ and cost $z^*(\boldsymbol{o}_k)$ are then computed by solving \eqref{eq:general solution}, where $z^*(\boldsymbol{o}_k)$ is the fitness of the chromosome $\boldsymbol{o}_k$. 
Following the initialization, the training phase starts with Elitism, which picks the top $\alpha\%$ of the fittest members from the current population and propagates them to the next generation. After that, an iterative procedure is performed to create offspring. In each iteration, two offloading orders are randomly picked from the current population $\boldsymbol{O}$ according to the probability distribution $\left\{\left.\frac{1/z^*(\boldsymbol{o}_k)}{\sum_{j=1}^P 1/z^*(\boldsymbol{o}_j)}\right\vert k\in\{1,2,\ldots,P\} \right\}$. Ordered crossover \cite{deep2011new} is then applied to create offspring. Subsequently, with a low probability $\beta$, mutation is performed to introduce diversity into the new population by shuffling individual offloading orders. The algorithm terminates upon meeting the stopping condition at which point it outputs the best solution found. In our simulations, we set the stopping condition for GA to be reaching a maximum number of generations, denoted as $G$.

\section{Simulation Studies} \label{sec:simulation}

In this section, we conduct simulation studies to evaluate the performance of the proposed  approaches. We start by describing the experiment setup in Sec. \ref{sec: Experiment Setting}. Next, we conduct two sets of studies to evaluate the optimality and efficiency of the proposed approaches in Sec. \ref{sec: optimality} and Sec. \ref{sec: efficiency}, respectively. We then investigate the impact of key parameters in Sec. \ref{sec:parameter_impact}, followed by an analysis of the effects of network characteristics.

\subsection{Experiment Setup} \label{sec: Experiment Setting}
We evaluate the proposed approaches on different network graphs generated using the method introduced in \cite {erdds1959random}. The network graphs are then transformed into tree topologies by using Dijkstra's algorithm. In each network topology, we configure the computing capacities $f_i$ of the servers by randomly generating values from the range of $[1, 10]$GHz. The values of the data transmission rates $R_{ij}$ are randomly generated from the range of $[10, 100]$Gbps. Moreover, we set $\gamma_i = 10^{-2}$, and $e_i = 30$dBm for all $i\in \mathcal{N}$. The task size is set to $Y = 1$ Gbits and $b = 10^6$ cycles/Gbit. All approaches are implemented in Python and evaluated on an Alienware Aurora R15 Gaming Desktop with a 12 Gen Intel i9 CPU and 64G of memory.

\subsection{Optimality Analysis} \label{sec: optimality}

We first evaluate the optimality of the two optimal approaches, CMO and PMO. For comparison, we implement the following four state-of-the-art distributed computing and computation offloading schemes as benchmarks:
\begin{itemize}
\item \textbf{Local computing (Local)}: In this approach, the master executes the entire task locally. 
\item \textbf{Partial offloading (Partial)}: In this approach, the master offloads part of the task to one of its one-hop neighbors. The offloading ratio and offloadee selection are optimized to minimize the task completion time. 
\item \textbf{Master-worker distributed computing (Master-worker)}: In this approach, the master distributes the task to its one-hop neighbors using the master-worker paradigm. The task allocation is optimized to minimize the task completion time. 
\item \textbf{Multi-hop offloading (Multi-hop)} \cite{gao2014opportunistic}: In this approach, the master offloads the whole task to the most powerful and reliable server in the network, which may be multiple hops away. 
\end{itemize}

Their performances are evaluated on four network graphs, which are transformed into trees with varying depths and breadths as illustrated in Fig. \ref{fig:topology}.

\begin{figure}[h]
    \centering
    \includegraphics[width=7cm]{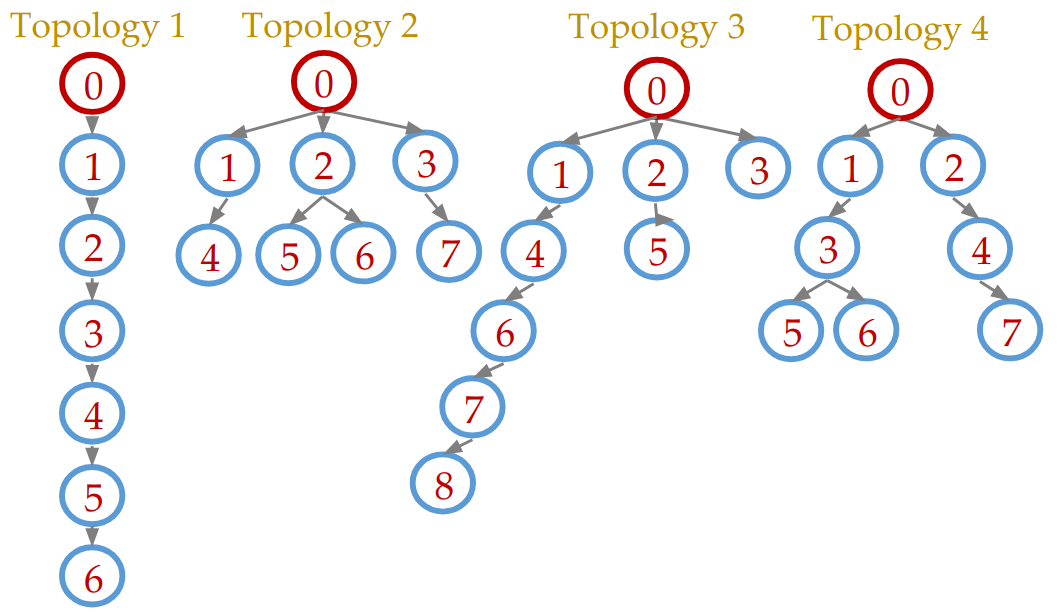}
    \caption{Network topologies evaluated in simulation studies.}
    \label{fig:topology}
\end{figure}

In the first experiment, we set the weights in the objective function to $w_1 = 1$ and $w_2 = 0$, which transforms the objective of our approach to minimize the task completion time only, just like the benchmarks. As shown in Fig.~\ref{fig:zero_weight}, our approaches outperform all benchmarks across all scenarios. Among the benchmarks, \textbf{Local} and \textbf{Multi-hop} have the poorest performance since they only use the computing resources from a single server. \textbf{Partial} outperforms local computing and \textbf{Multi-hop} by utilizing the resources from two servers. The \textbf{Master-worker}  achieves even better performance by utilizing computing resources from all servers within one hop. This experiment provides evidence that increasing the utilization of resources leads to better computing performance.    

In the second experiment, we randomly set the weights to $w_1=0.5$ and $w_2=0.05$, so that both computation efficiency and energy consumption are considered in our approaches. 
Note that these weight values are also used in the following experiments. Fig.~\ref{fig:cost} shows the comparison results, demonstrating the promising performance of our approaches in balancing task completion time and energy consumption. In Fig.~\ref{fig:time} and Fig.~\ref{fig:energy}, we show the task completion time, i.e., $\max_{i\in \mathcal{N}} T_i^{total}$, and the maximum energy consumption by any server, i.e., $\max_{i\in\mathcal{N}} E_i^{total}$, respectively. The result indicates that our approach outperforms all the benchmarks in both task completion time and maximum energy consumption.

\begin{figure}
  \centering
  \medskip
  \begin{subfigure}[t]{0.45\linewidth}
    \centering\includegraphics[width=1\linewidth]{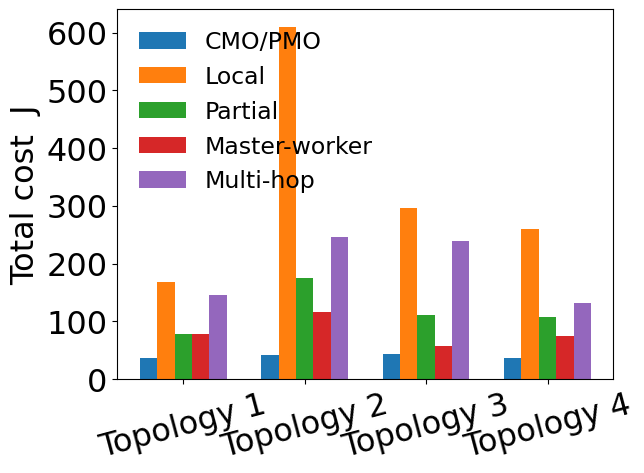}
    \caption{}
    \label{fig:zero_weight}
  \end{subfigure}
  \begin{subfigure}[t]{.45\linewidth}
    \centering\includegraphics[width=1\linewidth]{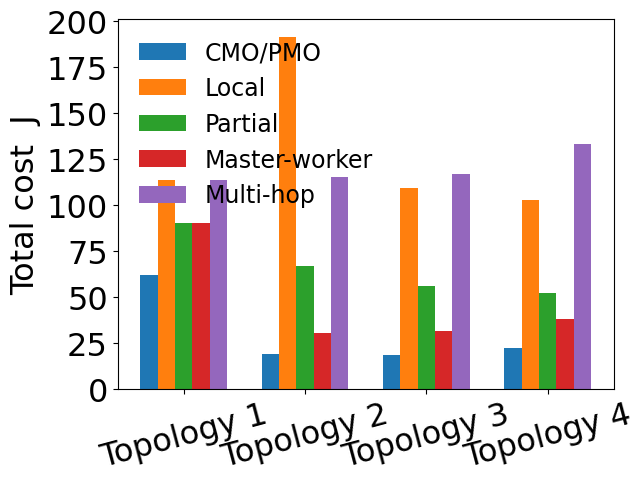}
    \caption{}
    \label{fig:cost}
  \end{subfigure}
  \caption{Total cost $J$ of different methods when considering (a) only time consumption; and (b) both time and energy consumption. }
\end{figure}

\begin{figure}[h]
  \centering
  \medskip
  \begin{subfigure}[t]{0.45\linewidth}
      \centering\includegraphics[width=1\linewidth]{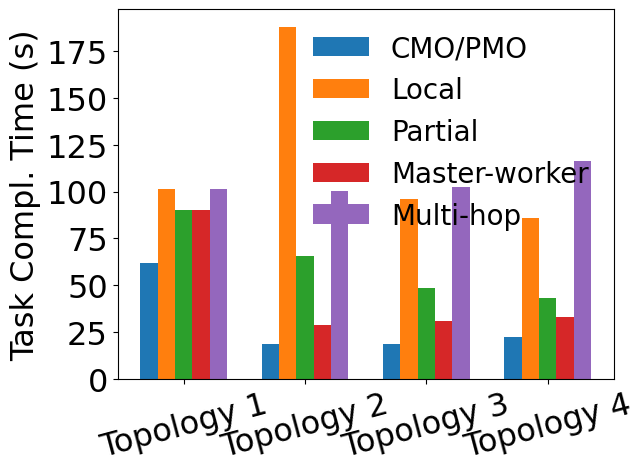}
    \caption{}
    \label{fig:time}
  \end{subfigure}
  \begin{subfigure}[t]{.45\linewidth}
    \centering\includegraphics[width=1\linewidth]{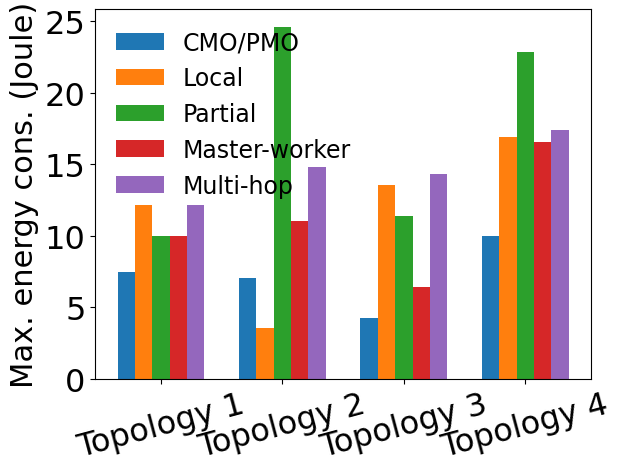}
    \caption{}
    \label{fig:energy}
  \end{subfigure}
  \caption{(a) Task completion time (b) Maximum energy consumption of different methods.} \label{energy and time 1} 
\end{figure}

\subsection{Efficiency Analysis} \label{sec: efficiency}

In this subsection, we evaluate the efficiency of the proposed optimal methods, CMO and PMO, as well as the proposed heuristic methods, NP, LP, and GA. For the implementation of the two worker selection methods, NP and LP, we first use them to prune the network tree, and then apply PMO to allocate tasks. GA is also implemented within the framework of PMO, and employed to determine the task allocation for each subtree, replacing CMO in Line 2 of Algorithm \ref{Algo:Greedy Flat Algorithm}.

\subsubsection{Small-Scale Networks}
We first consider the four small-scale network topologies depicted in Fig. \ref{fig:topology}. In this experiment, the threshold parameter $\theta_p$ in NP is set to $0.312, 0.43, 0.3, 0.4$,  respectively, such that one node in each topology is pruned. The threshold parameter $\xi$ in LP is set to $5, 1, 4, 2$ for the four topologies, respectively, resulting in the pruning of the last level of each topology.  For GA, the parameters are configured as $G=5, P=4, \alpha = 0.2, \beta=0.05$. 
To measure the efficiency of proposed approaches, we run each method 20 times and record the mean execution time, denoted as $T_{\text{exe}}$.

Fig. \ref{fig:Ourmethod_cost} shows the costs of the solutions found by the five methods. Comparing GA with the optimal methods, CMO and PMO, reveals that GA can find optimal solutions for small networks. This similarity in performance further demonstrates the optimality of GA for small networks. The worker selection methods, NP and LP, underperform compared to the other three methods, which is attributed to the reduced number of nodes involved in sharing the computational workload. Moreover, comparing the performance of LP across different topologies indicates that the extent of performance degradation is closely related to the proportion of nodes pruned from the network tree. Specifically, LP prunes 14.28\%, 57.14\%, 11.11\%, and 37.5\% of the nodes in the four topologies, respectively. The largest pruning proportion occur in  Topology 2, resulting in the 
maximum level of performance degradation. For NP, as only one node is ``pruned" in each topology, it performs better than LP in these scenarios.

The base-10 logarithm of the execution time, i.e., $\log_{10} T_{\text{exe}}$, of each method is shown in Fig. \ref{fig:Ourmethod_efficiency}. As expected, the optimal methods, CMO and PMO, are more time-consuming than the three heuristic methods. Moreover, PMO, being a parallelized version of CMO, significantly reduces execution time in Topologies 2-4 due to its parallelism. For Topology 1, since the root has a single subtree, PMO is equivalent to CMO. Among the heuristic methods, LP achieves the least execution time by pruning the most nodes and significantly reducing the search space. GA, on the other hand, is the least efficient and even underperforms PMO in Topologies 2 \& 4. This suggests that for small networks, PMO can be directly applied. Furthermore, comparing NP and PMO, we can observe that NP does not improve efficiency in all scenarios, despite reducing the number of workers. This is because only one node is ``pruned" in each topology, and the time saved by pruning is offset by the overhead generated by the pruning procedure.

\begin{figure}
  \centering
  \medskip
  \begin{subfigure}[t]{0.45\linewidth}
    \centering\includegraphics[width=1\linewidth]{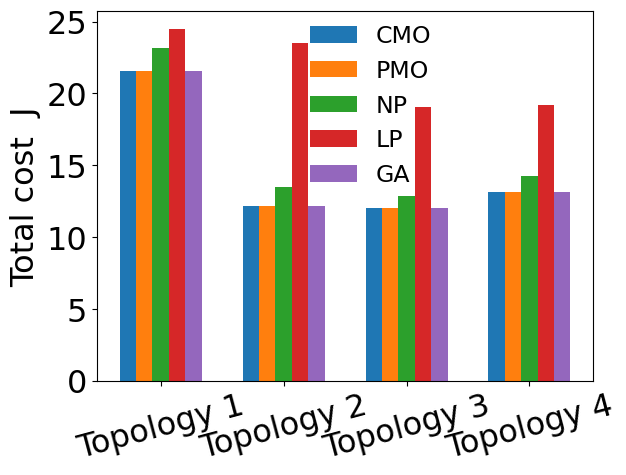}
    \caption{}
    \label{fig:Ourmethod_cost}
  \end{subfigure}
  \begin{subfigure}[t]{.45\linewidth}
    \centering\includegraphics[width=1\linewidth]{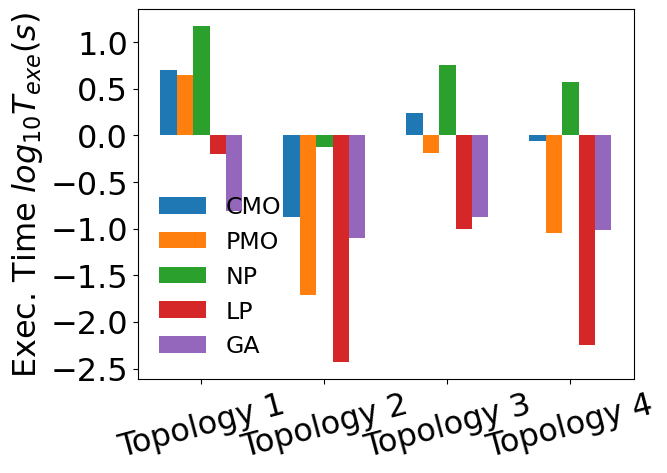}
    \caption{}
    \label{fig:Ourmethod_efficiency}
  \end{subfigure}
  \caption{(a) Total cost $J$ and (b) execution time of different methods for different network topologies.}
\end{figure}

As the performance of the proposed approaches largely depend on the network size, we further vary the network size by increasing the number of subtrees, $|\mathcal{I}_1|$, where each subtree consists of 2 levels and 1 node in each level. The scenario where the network size expands due to the growth of subtrees is explored in the subsequent subsection. 
In this experiment, we configure parameter $\theta_p$ in NP in a way such that one additional node is ``pruned" when including an additional subtree. Parameter $\xi$ in LP is set to $1$ in all cases, meaning that the node(s) in the last level are pruned. For GA, its parameters are configured as $P=4, G=100, \alpha = 0.2, \beta=0.05$. Fig. \ref{fig:NumOfSubtreesIncrease} shows the performance of different methods as the number of subtrees $|\mathcal{I}_1|$ increases. As we can see, increasing the number of subtrees results in a reduced total cost $J$, as more nodes are involved in sharing the computational workload. When comparing NP and LP, NP consistently outperforms LP. It's noteworthy that both methods prune the same number of nodes, each trimming one node from every subtree. This underscores the effectiveness of NP's worker selection process, which employs a more rigorous approach compared to LP that simply selects nodes at the top levels. However, the simplicity of LP makes it more efficient than NP, as shown in Fig. \ref{fig:NumOfSubtreesIncrease_efficiency}.
Additionally, from Fig. \ref{fig:NumOfSubtreesIncrease_efficiency}, we can observe a significant increase in the execution time of CMO as more subtrees are considered, compared to the other four methods. This is due to the parallelism inherent in the other four methods. Moreover, comparing PMO with the other methods further 
demonstrates the good performance of PMO in both optimality and efficiency in cases of small networks. 

\begin{figure}
  \centering
  \medskip
  \begin{subfigure}[t]{0.45\linewidth}
    \centering\includegraphics[width=1\linewidth]{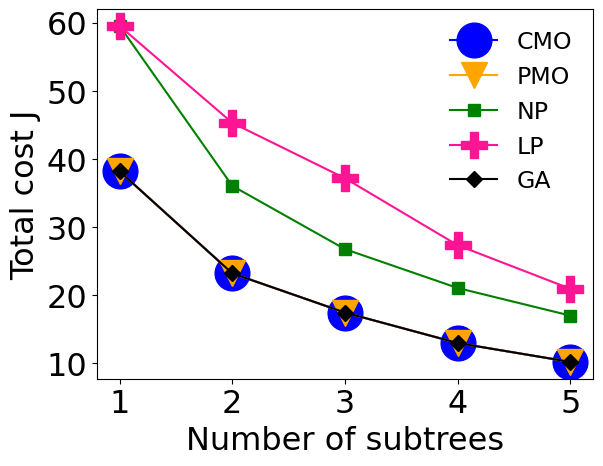}
    \caption{}
    \label{fig:NumOfSubtreesIncrease}
  \end{subfigure}
  \begin{subfigure}[t]{.48\linewidth}
    \centering\includegraphics[width=1\linewidth]{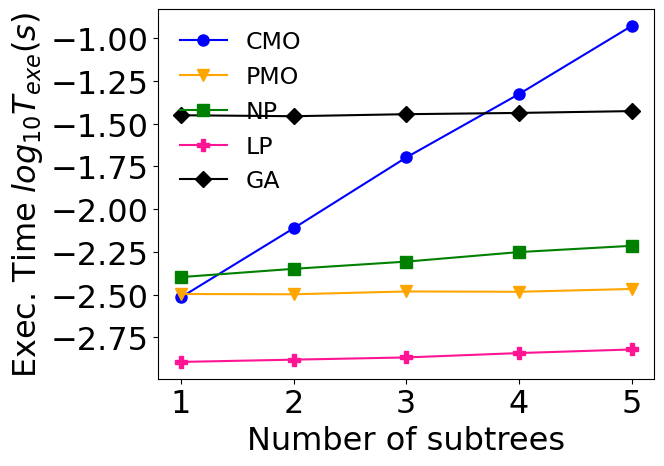}
    \caption{}
    \label{fig:NumOfSubtreesIncrease_efficiency}
  \end{subfigure}
  \caption{(a) Total cost $J$ and (b) execution time of different methods as the number of subtrees increases. } \label{fig:influence_of_subtrees}
\end{figure}

\subsubsection{Large-Scale Networks} \label{sec: GA strategy}
In this experiment, we consider larger networks and evaluate the performance of the three heuristic methods, NP, LP, and GA. For the implementation of NP and LP, GA is applied after pruning to determine the task allocation.  
Given that all these methods evaluate subtrees in parallel and their efficiency is bounded by the largest subtree, we evaluate their performance on networks with a single subtree. This approach allows us to avoid considering the impact of the number of subtrees. These networks are randomly generated with node counts of 10, 20, 30, and 50. The parameters in NP and LP are configured to prune a similar number of nodes as follows: the threshold $\theta_p$ in NP is set to $0.36, 0.45, 0.5, 0.38$, and the threshold $\xi$ in LP is set to $4, 3, 3, 12$ for the corresponding networks, respectively. The parameters in GA, applied across  all methods, are set to $P=4, G=100, \alpha = 0.2, \beta=0.05$.

As shown in Fig. \ref{fig:GA_cost}, the total cost $J$ generally decreases with the increase in network size for each method, as more nodes share the workload. Notably, the performance of GA degrades when the network size reaches 50. This is due to the large search space, making GA difficult to converge within 100 generations. Comparing the three methods, we can observe that GA outperforms the other two methods by considering all nodes in the network.  
NP generates better solutions than LP, although they prune roughly the same number of nodes. Additionally, NP and LP achieve performance comparable to GA in large networks (greater than 30 nodes), but with significantly lower execution times, as shown in Fig. \ref{fig:GA_efficiency}. This suggests that for large networks, NP and/or LP can be applied first to select a subset of workers, followed by GA for task allocation.

\begin{figure}
  \centering
  \medskip
  \begin{subfigure}[t]{0.45\linewidth}
    \centering\includegraphics[width=1\linewidth]{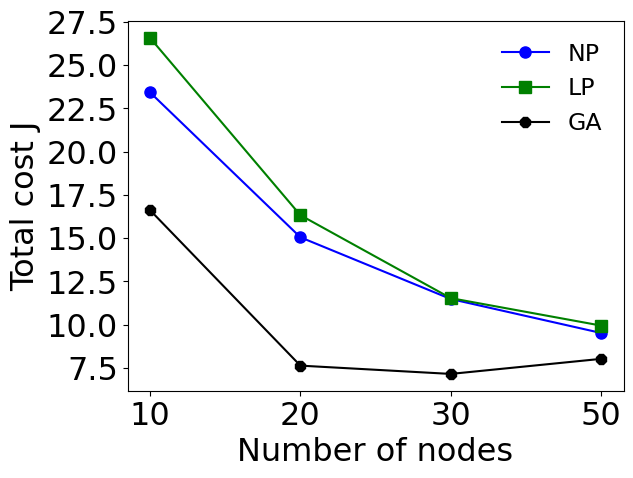}
    \caption{}
    \label{fig:GA_cost}
  \end{subfigure}
  \begin{subfigure}[t]{.45\linewidth}
    \centering\includegraphics[width=1\linewidth]{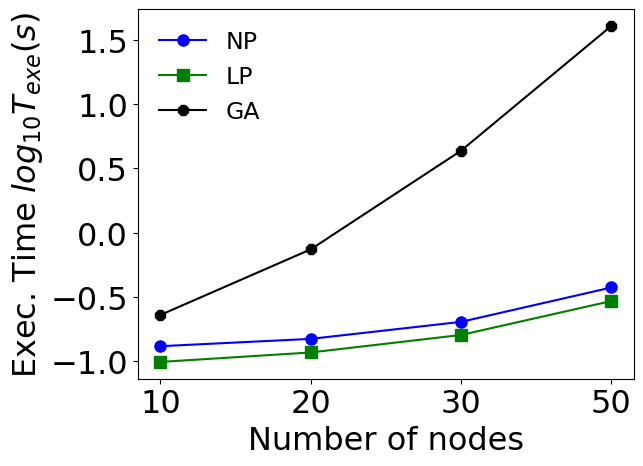}
    \caption{}
    \label{fig:GA_efficiency}
  \end{subfigure}
  \caption{(a) Total cost $J$ and (b) execution time of different heuristic methods as the number of nodes in the network increases. }
  \label{fig:GA performance}
\end{figure}

\subsection{Parameter Impact Analysis} \label{sec:parameter_impact}
In this subsection, we investigate the impact of key parameters in the proposed heuristic methods, including (1) the threshold $\theta_p$ in NP, (2) threshold $\xi$ in LP.
All experiments are conducted on the network with 20 nodes, as described in Sec. \ref{sec: GA strategy}. GA is employed for task allocation, using the same configuration detailed in Sec. \ref{sec: GA strategy}.

\subsubsection{Threshold $\theta_p$}
In NP, a worker is selected if the cost reduction from including this node exceeds the threshold $\theta_p$. Therefore, a higher threshold will result in fewer nodes being selected and more nodes being ``pruned". This is demonstrated by the results shown in Fig. \ref{fig:Node_compare_cost}. As we can see, the best performance is achieved when $\theta_p = 0$, in which case no nodes are ``pruned". The worst performance occurs at $\theta_p = 1$, where all workers are ``pruned" and all computations are done locally at the master. Moreover, as $\theta_p$ decreases, more workers are selected, resulting in a decrease in cost $J$ (see Fig. \ref{fig:Node_compare_cost}) but an increase in execution time $T_{\text{exe}}$ (see Fig. \ref{fig:Node_compare_time}). Notably, when $\theta$ is reduced to 0.4, $J$ tends to converge. This suggests that an appropriate value of $\theta_p$ that balances optimality and efficiency can be identified by selecting the value at which a sharp change in cost $J$ occurs.

\begin{figure}
  \centering
  \medskip
  \begin{subfigure}[t]{0.45\linewidth}
    \centering\includegraphics[width=1\linewidth]{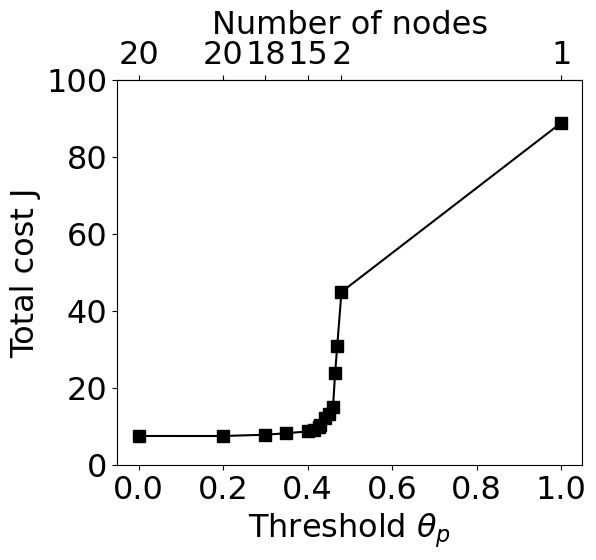}
    \caption{}
    \label{fig:Node_compare_cost}
  \end{subfigure}
  \begin{subfigure}[t]{.45\linewidth}
    \centering\includegraphics[width=1\linewidth]{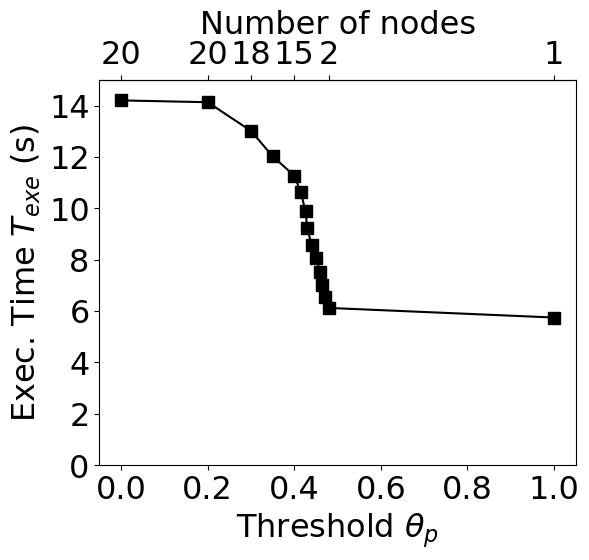}
    \caption{}
    \label{fig:Node_compare_time}
  \end{subfigure}
  \caption{(a) Total cost $J$ and (b) execution time of NP as the threshold $\theta_p$ increases. The upper $x$-axes show the number of nodes selected as workers.  }
  \label{fig:Node_threshold}
\end{figure}

\subsubsection{Threshold $\xi$}
LP selects all nodes in the top $\xi$ levels as workers.  
In the special case when $\xi=0$, no nodes are selected, resulting in all computations being conducted locally at the master. This leads to the highest cost $J$, as shown in Fig. \ref{fig:Experiment_B_redraw_cost}.  
As $\xi$ increases and more nodes are selected, performance improves, as indicated by the decreasing cost $J$. However, the performance improvement slows down when $\xi$ exceeds 3. The best performance is achieved when $\xi$ reaches its maximum value, $H$ (height of the network tree), which is 12 in this experiment. Given the rapid increase in execution time with higher $\xi$, as shown in Fig. \ref{fig:Experiment_B_redraw_time}, a proper value for $\xi$ can be chosen at the point where the rate of decrease in $J$ slows down.

\begin{figure}
  \centering
  \medskip
  \begin{subfigure}[t]{0.45\linewidth}
    \centering\includegraphics[width=1\linewidth]{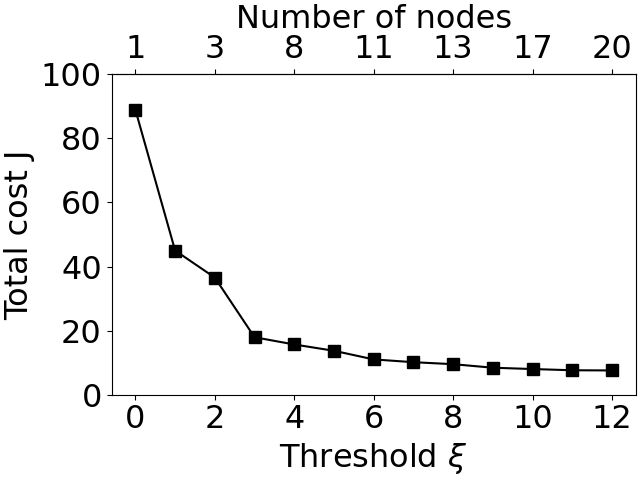}
    \caption{}
    \label{fig:Experiment_B_redraw_cost}
  \end{subfigure}
  \begin{subfigure}[t]{.45\linewidth}
    \centering\includegraphics[width=1\linewidth]{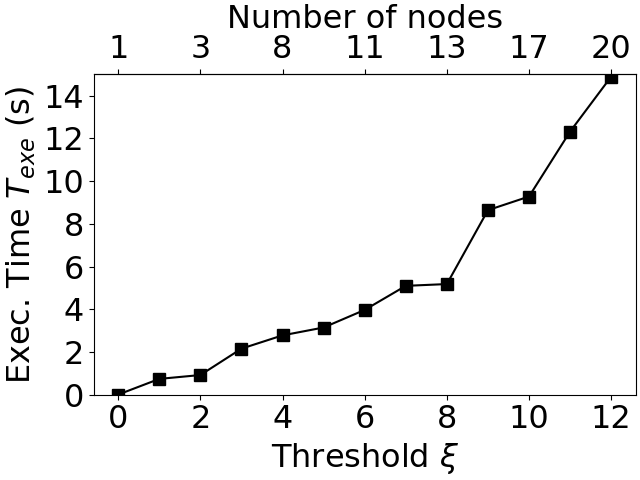}
    \caption{}
    \label{fig:Experiment_B_redraw_time}
  \end{subfigure}
  \caption{(a) Total cost $J$ and (b) execution time of LP as the threshold $\xi$ increases. } \label{fig:Level-forward analysis}
\end{figure}

\subsection{Impact of Network Characteristics} 
The optimal task distribution decisions highly rely on the network characteristics. In this subsection, we explore how communication and computing parameters, specifically $R_{i,j}$ and $f_i$, affect these decisions. For this analysis, we focus on Topology 4, as shown in 4 in Fig. \ref{fig:topology}. 

In the first experiment, we vary $R_{0,1}$, which represents the communication capacity between the master node (Node 0) and its left child (Node 1), from 0.3 Gbps to 10 Gbps. All other settings remain consistent with the previous studies. The optimal task allocation derived by PMO is shown in  Fig.\ref{fig:comm_impact}. As the figure demonstrates, when $R_{0,1}$ is small, communication becomes a bottleneck, preventing the master from offloading tasks to Node 1 or to its descendants. However, as $R_{0,1}$ increases, more workload is offloaded to Node 1. Once $R_{0,1}$ exceeds certain thresholds, tasks are also offloaded to Node 1's children and even grandchildren. With more nodes contributing to workload distribution, nodes in the right subtree of the master begin to receive fewer tasks. This study suggests that if a communication link is too slow, both the connected downstream node and its descendants may be pruned from the topology before executing CMO/PMO.
To understand the impact of the computing characteristic, we instead vary the computing power of Node 1, $f_1$, from 0.022 GHz to 21 GHz. As shown in Fig. \ref{fig:comp_impact}, the master starts to offload tasks to Node 1 when its computing power exceeds a certain threshold. Moreover, when it shares more workload, the workloads assigned to all other nodes decrease simultaneously. 

Notably, the network characteristics determine whether tasks are offloaded to a node, regardless of the total task size $Y$, as inferred from Theorem \ref{Thm:factor}. To demonstrate this, we vary $Y$ while keeping the network characteristics constant. Table 1 summarizes the optimal task allocations and the corresponding total costs computed by PMO. As shown, when $Y$ increases, both the workload assigned to each node $y_i$ and the total cost $J$ rise proportionally.

\vspace{-0.45cm}

\begin{figure}[h]
  \centering
  \medskip
  \begin{subfigure}[t]{0.45\linewidth}
    \centering\includegraphics[width=1\linewidth]{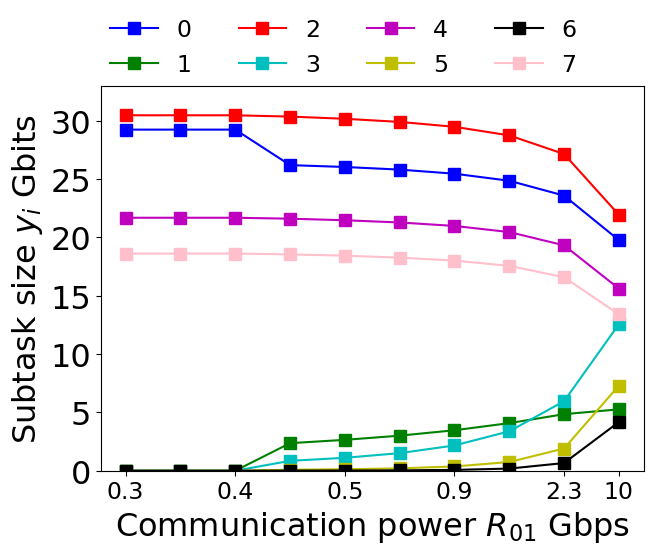}
    \caption{}
    \label{fig:comm_impact}
  \end{subfigure}
  \begin{subfigure}[t]{.45\linewidth}
    \centering\includegraphics[width=1\linewidth]{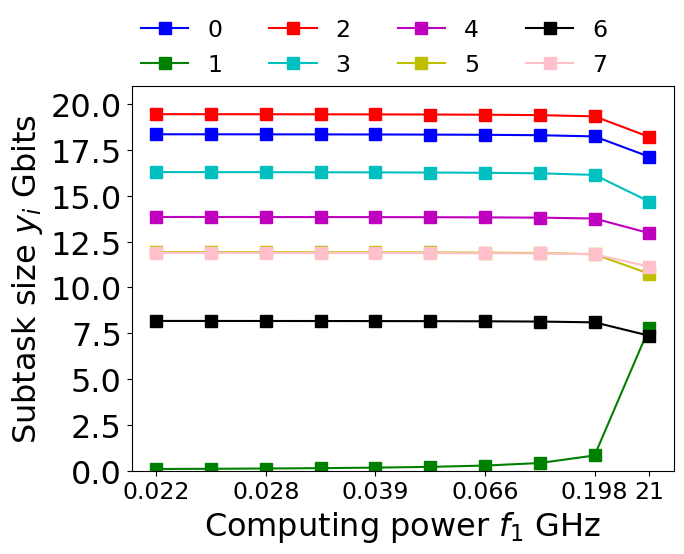}
    \caption{}
    \label{fig:comp_impact}
  \end{subfigure}
  \caption{Optimal task allocations at different values of (a) $R_{0,1}$ and (b) $f_1$.} \label{fig:offloading decision} 
\end{figure}





\begin{table}
  \includegraphics[width=\linewidth]{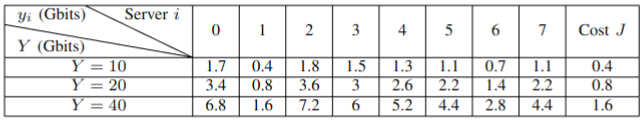}
\caption{Optimal task allocation and total cost for different values of $Y$}  \label{table: task size}
\end{table}

\section{Conclusion and Future Works}
\label{sec:conclusion}
This paper introduces a novel multi-layered distributed computing framework that expands the computing capabilities of networked computing systems. 
Unlike conventional distributed computing paradigms that limit resource sharing to one-hop neighborhoods, our framework explores layered network structures  to enable resource sharing beyond one-hop neighborhoods, effectively utilizing the resources of the entire network. To optimize system performance, we formulated an MIP problem that jointly minimizes task completion time and energy consumption through optimal task allocation and scheduling. Two exact methods, CMO and PMO, were proposed to solve this problem optimally, with PMO enhancing CMO's efficiency by exploiting the parallelism in task distribution across the master's subtrees. We also introduced an offline-online computation scheme to enable the real-time execution of CMO and PMO and allow them to handle dynamic networks with time-varying characteristics. To further enhance efficiency and scalability, three heuristic methods were introduced, including NP and LP for reducing the number of workers and GA for efficiently finding (sub-)optimal solutions. Simulation results demonstrate the superiority of our approaches over existing distributed computing and computation offloading schemes. Moreover, PMO shows promising performance in both optimality and efficiency for small networks. For larger networks, the results suggest applying NP or LP to reduce workers before using GA or PMO for task allocation. The results also show that NP outperforms LP in terms of optimality but is less efficient due to its more rigorous worker selection process. Additionally, studies on the impact of NP's and LP's parameters offer insights into their configurations. Lastly, the analysis of network characteristics highlights how the communication and computing capacities of individual servers influence task distribution decisions. 
In the future, we will extend this work to consider multi-task scenarios as well as dynamic and mobile networks. We will also investigate the hierarchical master-work paradigm.

\section*{Acknowledgment}
We would like to thank National Science Foundation under Grant CAREER-2048266 and CCF-2402689 for the support of this work.

\printbibliography

\begin{IEEEbiography}[{\includegraphics[width=1in,height=1.25in,clip,keepaspectratio]{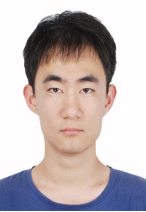}}]{Ke Ma} is currently a PhD candidate in the joint doctoral program of University of California, San Diego and San Diego State University. He received the B.Sc. degree in Communication Engineer from Donghua University, Shanghai, China, in 2019 and the M.S. degree in Computer Engineer from University of California at Riverside, California, in 2022.\end{IEEEbiography}

\begin{IEEEbiography}
[{\includegraphics[width=1in,height=1.25in,clip,keepaspectratio]{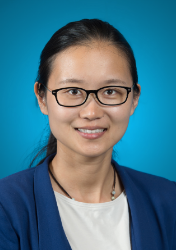}}]{Junfei Xie}(S'13-M'16-SM'21) is currently an Associate Professor in the Department of Electrical and Computer Engineering at San Diego State University. She received the B.S. degree in Electrical Engineering from University of Electronic Science and Technology of China (UESTC), Chengdu, China, in 2012. She received the M.S. degree in Electrical Engineering in 2013 and the Ph.D. degree in Computer Science and Engineering from University of North Texas, Denton, TX, in 2016. From 2016 to 2019, she was an Assistant Professor in the Department of Computing Sciences at Texas A\&M University-Corpus Christi. She is the recipient of the NSF CAREER Award. Her current research interests include large-scale dynamic system design and control, airborne networks, airborne computing, and air traffic flow management, etc.\end{IEEEbiography}

\end{document}